\documentclass [a4paper,11pt,onecolumn]{article}
\usepackage{amsmath,dsfont,yfonts,amssymb,epsfig,psfrag,amsthm,bm,multirow, graphicx,color}
\usepackage{appendix}
\DeclareMathOperator*{\argmax}{arg\,max}

\newtheorem{theorem}{Theorem}
\newtheorem{proposition}[theorem]{Proposition}
\newtheorem{corollary}[theorem]{Corollary}

\newtheorem{lemma}[theorem]{Lemma}

  \def\marginset#1#2{                      
  \setlength{\oddsidemargin}{#1}
    \setlength{\evensidemargin}{0mm}

  \setlength{\hoffset}{\paperwidth}
  \addtolength{\hoffset}{-\oddsidemargin}
  \addtolength{\hoffset}{-\textwidth}
  \addtolength{\hoffset}{-\evensidemargin}
  \setlength{\hoffset}{0.5\hoffset}
  \addtolength{\hoffset}{-1in}           

    \setlength{\voffset}{-1in}             
  \setlength{\topmargin}{\paperheight}
  \addtolength{\topmargin}{-\headheight}
  \addtolength{\topmargin}{-\headsep}
  \addtolength{\topmargin}{-\textheight}
  \addtolength{\topmargin}{-\footskip}
  \addtolength{\topmargin}{#2}
  \setlength{\topmargin}{0.4\topmargin} }

  \marginset{0mm}{0mm}

\makeatletter
\def\@maketitle{%
 \newpage
 \begin{center}%
 \let \footnote \thanks
   {\LARGE \@title \par}%
   \vskip 1em%
   {\large
     \lineskip .5em%
     \begin{tabular}[t]{c}%
       \@author
     \end{tabular}\par}%
 \end{center}%
 \par
 \vskip 1.5em}
\makeatother

\begin{document}

\title{Scheduling with Rate Adaptation under Incomplete Knowledge of Channel/Estimator Statistics}
\author{\emph{Wenzhuo Ouyang, Sugumar Murugesan, Atilla Eryilmaz, Ness B.
Shroff}\\
\{ouyangw, murugess, eryilmaz, shroff\}@ece.osu.edu\\
Department of Electrical and Computer Engineering\\
The Ohio State University\\
Columbus, OH, 43210 } \maketitle

\begin{abstract}
In time-varying wireless networks, the states of the communication
channels are subject to random variations, and hence need to be
estimated for efficient rate adaptation and scheduling. The
estimation mechanism possesses inaccuracies that need to be tackled
in a probabilistic framework. In this work, we study scheduling with
rate adaptation in single-hop queueing networks under two levels of
channel uncertainty: when the channel estimates are inaccurate but
complete knowledge of the channel/estimator joint statistics is
available at the scheduler; and when the knowledge of the joint
statistics is incomplete. In the former case, we characterize the
network stability region and show that a maximum-weight type
scheduling policy is throughput-optimal. In the latter case, we
propose a joint channel statistics learning - scheduling policy. With an
associated trade-off in average packet delay and convergence time,
the proposed policy has a stability region arbitrarily close to the
stability region of the network under full knowledge of
channel/estimator joint statistics.
 \end{abstract}

\section{Introduction}

Scheduling in wireless networks is a critical component of resource
allocation that aims to maximize the overall network utility subject
to link interference and queue stability constraints. Since the
seminal paper by Tassiulas and Ephremides (\cite{backpressure}),
maximum-weight type algorithms have been intensely studied
(e.g., \cite{MWM}-\cite{Neely05}) and found to be throughput-optimal in
various network settings. The majority of existing works employing
maximum-weight type schedulers are based on the assumption that full
knowledge of channel state information (CSI) is available at the
scheduler. In realistic scenarios, however, due to random variations
in the channel, full CSI is rarely, if ever, available at the
scheduler. The dynamics of the scheduling problem with imperfect CSI
is, therefore, vastly different from the problem with full CSI
in the following two ways (1) a non-trivial amount of
network resource, that could otherwise be used for data
transmission, is spent in learning the channel; (2) the acquired
information on the channel is potentially inaccurate, essentially
underscoring the need for intelligent rate adaptation and user
scheduling. Realistic networks are thus characterized by a convolved
interplay between channel estimation, rate adaptation, and multiuser
scheduling mechanisms.

These complicated dynamics are studied under various network
settings in recent works (\cite{infreq}-\cite{2stage}). In
\cite{infreq}, the authors study scheduling in single-hop wireless
networks with Markov-modeled binary ON-OFF channels. Here scheduling
decisions are made based on cost-free estimates of the channel
obtained once every few slots. The authors show that a
maximum-weight type scheduling policy, that takes into account the
probabilistic inaccuracy in the channel estimates and the memory in
the Markovian channel, is throughput-optimal. In \cite{Topology},
the authors study decentralized scheduling under partial CSI in
multi-hop wireless networks with Markov-modeled channels. Here, each
user knows its channel perfectly and has access to delayed CSI of
other users' channels. The authors characterize the stability region
of the network and show that a maximum-weight type threshold policy,
implemented in a decentralized fashion at each user, is throughput
optimal.

In \cite{Neely_CDC}, the authors study scheduling under imperfect
CSI in single-hop networks with independent and identically
distributed (\textit{i.i.d.}) channels. They consider a two-stage
decision setup: in the first stage, the scheduler decides whether to
estimate the channel with a corresponding energy cost; in the second
stage, scheduling with rate adaptation is performed based on the
outcome of the first stage. Under this setup, the authors propose a
maximum-weight type scheduling policy that minimizes the energy
consumption subject to queue stability.

While studying scheduling under imperfect CSI is a first step in the
right direction, these works assume that complete knowledge of the
channel/estimator joint statistics, which is crucial for the success
of opportunistic scheduling, is readily available at the scheduler.
This is another simplifying assumption that need not always hold in
reality. Taking note of this, we study scheduling in single-hop
networks under imperfect CSI, and when the knowledge of the
channel/estimator joint statistics is incomplete at the scheduler.
We propose a joint statistics learning-scheduling policy that
allocates a fraction of the time slots (the exploration slots) to
continuously learn the channel/estimator statistics, which in turn
is used for scheduling and rate adaptation during data transmission
slots. Note that our setup is similar to the setup considered in
\cite{2stage}. Here the author considers a two-stage decision setup.
When applied to the scheduling problem, this work can be interpreted
as follows. One of $K$ estimators is chosen to estimate the channel
in the first stage, with unknown channel/estimator joint statistics.
The second stage decision is made to minimize a \textit{known}
function of the estimate obtained in the first stage. Our problem is
different from this setup in that the channel/estimator joint
statistics is important to optimize the second stage decision in
our problem - i.e., scheduling with rate adaptation. This is not the
case in \cite{2stage} where a \textit{known} function of the
estimate is optimized and the channel/estimator joint statistics is
helpful only in the first stage that decides one of $K$ estimators.
Our contribution is two-fold:
\begin{itemize}
\item
When complete knowledge of the channel/estimator joint statistics is
available at the scheduler, we characterize the network stability
region and show that a simple maximum-weight type scheduling policy
is throughput-optimal. It is worth contrasting this result with
those in \cite{infreq}-\cite{Neely_CDC}. In these works,
imperfection of CSI is assumed to be caused by \textit{specific}
factors like delayed channel feedback, infrequent channel
measurement, etc, whereas, in our model, since the channel/estimator
joint statistics is unconstrained, the CSI inaccuracy is captured
in a more general probabilistic framework.

\item Using the preceding system level results as a benchmark, we study scheduling
under incomplete knowledge of the channel/estimator joint
statistics. We propose a scheduling policy with an in-built
statistics learning mechanism and show that, with a corresponding
trade-off in the average packet delay before convergence, the
stability region of the proposed policy can be pushed arbitrarily
close to the network stability region under full knowledge of
channel/estimator statistics.
\end{itemize}

The paper is organized as follows. Section II formalizes the system
model. In Section III, we characterize the stability region of the
network and propose a throughput-optimal scheduling policy. In
Section IV, we study joint statistics learning-scheduling and rate
adaptation when the scheduler has incomplete knowledge of
channel/estimator statistics. Concluding remarks are provided in
Section V.

\section{System Model}

We consider a wireless downlink communication scenario with one base
station and $N$ mobile users. Data packets to be transmitted from
the base station to the users are stored in $N$ separate queues at
the base station. Time is slotted with the slots of all the users
synchronized. The channel between the base station and each user is
\textit{i.i.d.} across time slots and independent across users. We do
not assign any specific distribution to the channels throughout this
work. The channel state of a user in a slot denotes the number of
packets that can be successfully transmitted without outage to that
user, in that slot. Transmission at a rate below the channel state
always succeeds, while transmission at a rate above the channel
state always fails. We assume the channel state lies in a finite
discrete state space $\mathcal{S}$. Let $C_i[t]$ be the random
variable denoting the channel state of user $i$ in slot $t$. The
channel state of the network in slot $t$ is denoted by the vector
$\bm C[t]=\big [C_1[t], C_2[t], \cdots, C_N[t] \big]\in
\mathcal{S}^N$. In each slot, the scheduler has access to estimates
of the channel states, i.e., $\widehat{\bm C}[t]=\big
[\widehat{C}_1[t], \widehat{C}_2[t], \cdots, \widehat{C}_N[t]
\big]\in \mathcal{S}^N$. The estimator is fixed for each user and
the estimates are independent across users. The channel/estimator joint statistics for user $i$ is given by the
$|\mathcal{S}|^2$ probabilities $P(C_i{=}c_i,
\widehat{C}_i{=}\hat{c}_i)$, $\forall c_i{\in}\mathcal{S},
\hat{c}_i{\in}\mathcal{S}$.

We adopt the one-hop interference model, where, in each slot, only
one user is scheduled for data transmission. The scheduler (base
station), based on the channel estimate and the queue length
information, decides which user to schedule and performs rate
adaptation in order to maximize the overall network stability
region. Let $I[t]$ and $R[t]$ denote the index of the user scheduled
to transmit and the corresponding rate of transmission,
respectively, at slot $t$. Due to potential mismatch between the
channel estimates and the actual channels, it is possible that the
allocated rate is larger than the actual channel rate, thus leading
to outage. In this case, the packet is retained at the head of the
queue and a retransmission will be attempted later. Let $Q_i[t]$
denote the state (length) of queue $i$ at the beginning of slot $t$.
Let $A_i[t]$ denote the number of exogenous packet arrivals at queue
$i$ at the beginning of slot $t$ with $E[A_i[t]]=\lambda_i$. The
queue state evolution can now be written as a discrete stochastic
process:

\vspace{-11pt} {\small
\begin{equation}
Q_i[t{+}1]=\big[ Q_i[t]{-}\bm 1 (I[t]{=}i) R[t] \cdot \bm 1(R[t]
{\leq} C_i[t])\big]^+ {+} A_i [t],
\end{equation}
} \hspace{-5pt} where $[\cdot ]^+= \max \{ 0, \cdot \}$.
\vspace{2pt} We adopt the following definition of queue stability
\cite{MWM}: Queue $i$ is stable if there exists a limiting
stationary distribution $F_i$ such that $\lim_{t\rightarrow \infty}
P( Q_i[t] \leq q)= F_i(q)$.

\vspace{-8pt}
\section{Full Knowledge of Channel/Estimator Joint Statistics}

In this section, we consider the scenario where the scheduler has
full access to the channel/estimator joint statistics, i.e.,
$P(C_i{=}c_i, \widehat{C}_i{=}\hat{c}_i)$, $\forall c_i{\in}\mathcal{S},
\hat{c}_i{\in}\mathcal{S}$ for $i{\in}\{1,{\ldots},N\}$. We
characterize the network stability region next. \vspace{-14pt}

\subsection{Network Stability Region}

Consider the class of stationary scheduling policies $G$ that base
their decision on the current queue length information $[Q_1, \ldots
, Q_N ]$, the channel estimates $[\widehat{C}_i,\ldots,
\widehat{C}_N]$, and full knowledge of channel/estimator joint
statistics. Define the network stability region as the closure of
the arrival rates that can be supported by the policies in $G$
without leading to system instability. Let $P_{\widehat{ \bm
C}}(\hat{\bm c}=[\hat{c}_1,\ldots,\hat{c}_N])$ denote the
probability of the channel estimate vector. Thus, \vspace{-6pt}
\begin{equation}
P_{\widehat{ \bm C}}(\hat{\bm
c}=[\hat{c}_1,\ldots,\hat{c}_N])=\prod_{i=1}^N
P(\widehat{C}_i=\hat{c}_i),
\end{equation}
where the probabilities $P(\widehat{C}_i=\hat{c}_i)$ are evaluated
from the knowledge of the channel/estimator joint statistics.
Defining $\mathcal{CH}[\mathcal{A}]$ as the convex hull (\cite{Boyd})
of set $\mathcal{A}$ and $\vec{\bm 1}_i$ as the $i^\textrm{th}$
coordinate vector, we record our result on the network stability
region below.
\begin{proposition}
\label{stability region chara} The stability region of the network
is given by
\begin{eqnarray}
\bm \Lambda = \sum_{\hat{\bm c}\in \mathcal{S}^N} P_{\widehat{ \bm
C}}(\hat{\bm c}) \cdot \mathcal{CH}\Big[\ \bm 0, P(C_i{\geq} r_i^*(\hat{c}_i)  \big
|\widehat{C}_i{=}\hat{c}_i\ ) r_i^*(\hat{c}_i) \cdot \vec{\bm 1}_i;
i=1,\cdots, N\Big],\nonumber
\end{eqnarray}
\hspace{-8pt} where $r_{i}^*(\hat{c}_i)= \argmax_{r\in \mathcal{S}}
\big \{P(C_i\geq r \Big |\widehat{C}_i=\hat{c}_i\ ) \cdot r \big \}$
and the conditional probabilities $P(C_i\geq
r_i^*(\hat{c}_i)|\widehat{C}_i=\hat{c}_i)$ are evaluated from the
knowledge of the channel/estimator joint statistics.
\end{proposition}

\textit{Proof Outline:} The proof contains two parts. We first show
that any rate vector $\bm \lambda$ strictly within $\bm \Lambda$ is
stably supportable by some randomized stationary policy. In the
second part, we establish that any arrival rate $\bm \lambda$
outside $\bm \Lambda$ is not supportable by any policy. We show this
by first identifying a hyperplane that separates $\bm \lambda$ and
$\bm \Lambda$ using the strict separation theorem (\cite{Separation}).
We then define an appropriate Lyapunov function and show that, for
any scheduling policy, there exists a positive drift, thus rendering
the queues unstable (\cite{Sean}). Details of the proof are available
in Appendix A.
\vspace{3pt}

\subsection{Optimal Scheduling and Rate Allocation}

In this section, we propose a maximum-weight type scheduling policy
with rate adaptation and show that it is throughput-optimal, i.e.,
it can support any arrival rate that can be supported by any other
policy in $G$. The policy is introduced next.\vspace{12pt}\\
\fbox{
\parbox[l]{0.94\linewidth}{
\vspace{9pt}\textbf{Scheduling Policy $\Psi$}\\

\vspace{-10pt}
At time slot $t$, the base station makes the scheduling and rate
adaptation decisions based on the channel/estimator joint statistics
and the channel estimate vector $\widehat{\bm C}=\hat{\bm c}$
(the time index is dropped for notational simplicity).\\

\emph{(1) Rate Adaptation:}\\

\quad \quad For each user $i$, assign rate $R_i$ such that,\\
\begin{displaymath}
R_i= \argmax_{r \in \mathcal{S}} \big \{P(C_i\geq r  \big
|\widehat{C}_i=\hat{c}_i) \cdot r \big\}
\end{displaymath}

\emph{(2) Scheduling Decision:}\\

\quad \quad Schedule the user $I$ that maximizes the
queue-weighted rate $R_i$, as follows:
\begin{displaymath}
I=\argmax_i \big \{Q_i \cdot P(C_i\geq R_i  \big
|\widehat{C}_i=\hat{c}_i) \cdot R_i\big \}
\end{displaymath}
} }\\ \\

\noindent Note that when the channel state estimation is accurate,
the conditional probability $P(C_i\geq r  \big
|\widehat{C}_i=\hat{c}_i)$ will be a step function, with $\Psi$
essentially becoming the classic maximum-weight policy in \cite{MWM}.
The next proposition establishes the throughput optimality of policy
$\Psi$. Details are provided in Appendix B.
\begin{proposition}
\label{stability} The scheduling policy $\Psi$ stably supports all arrival
rates that lie in the interior of the stability region $\bm
\Lambda$.
\end{proposition}

\vspace{1pt}

\textit{Proof Outline:}
The proof proceeds as follows. Consider a
Lyapunov function $L(\bm Q[t])=\sum_{i=1}^N Q_i^2[t]$. For any
arrival rate $\bm \lambda$ that lies strictly within the stability
region $\bm \Lambda$, we know it is stably supportable by some
policy $G_0$. Under $G_0$, we show that the corresponding Lyapunov drift is negative. We then show that policy $\Psi$
minimizes the Lyapunov drift and hence it will have a negative
drift, thus establishing the throughput optimality of $\Psi$.
\vspace{5pt}

The results obtained thus far when the channel/estimator joint
statistics is available at the scheduler are along expected lines.
Nonetheless, they serve as a benchmark to the rest of the work under
incomplete knowledge of the channel/estimator joint statistics,
which is \emph{the main focus of the paper}.



\begin{figure*}
\centering
\includegraphics[width=3.2in]{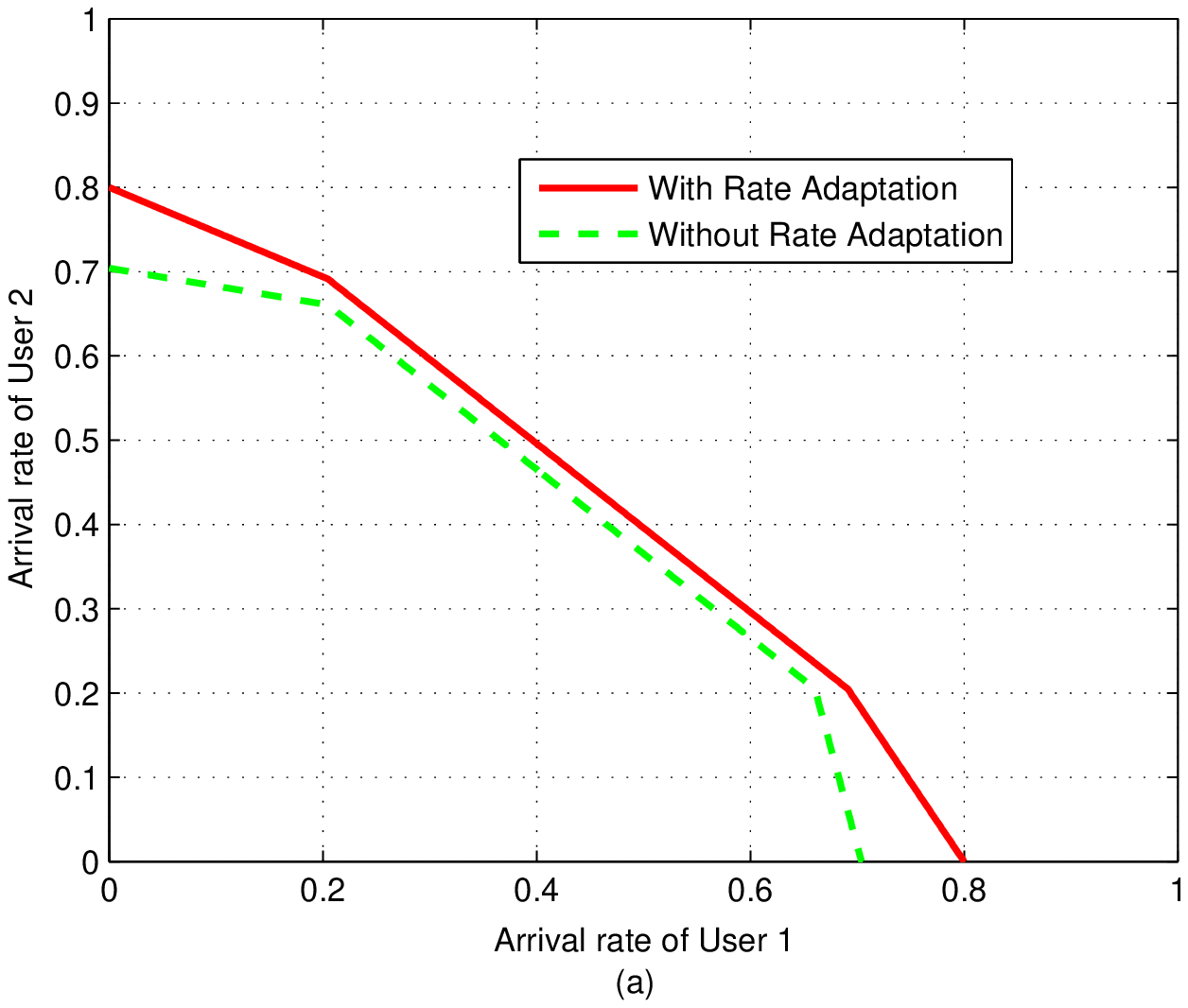}
\includegraphics[width=3.2in]{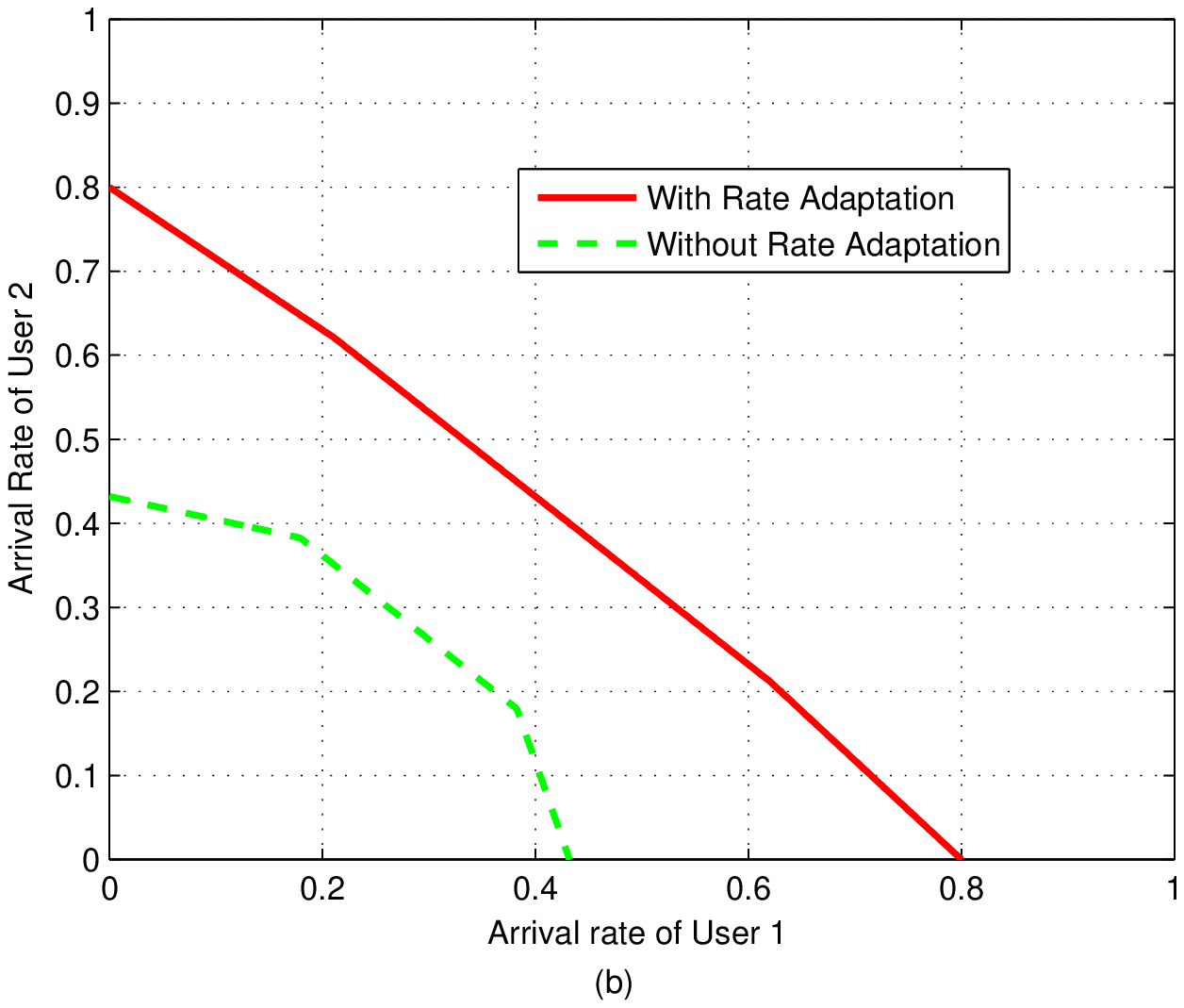}
\caption{Illustration of the system level gains associated with
joint scheduling and rate control. $P_i(\widehat{C}_i=k\big |
C_i=k)$ denotes the probability that the channel estimate of user
$i$ is $k$ given the actual channel of user $i$ is $k$. (a)
$P_i(\widehat{C}_i=k\big | C_i=k)=0.8$, \ $k \in \{ 0.2, 1 \}, \
i\in \{1,2\}$; (b) $P_i(\widehat{C}_i=k\big | C_i=k)=0.4$, $k \in \{
0.2, 1 \}, \ i\in \{1,2\}$.} \label{fig:example}
\end{figure*}

\section{Incomplete Knowledge of Channel/Estimator Joint Statistics}

In this section, we study scheduling with rate adaptation when the
scheduler only has knowledge of the marginal statistics of the
estimator, i.e., $P(\widehat{C}_i=\hat{c}_i)$, $\forall
c_i\in\mathcal{S}$, $i\in\{1,\ldots,N\}$, and hence, the knowledge
of the channel/estimator joint statistics is incomplete at the
scheduler. We first illustrate, with a simple example, that
significant system level losses are incurred when no effort is made
to learn these statistics, and hence no rate adaptation is
performed.

\subsection{Illustration of the Gains from Rate Adaptation}

With incomplete information on the channel-estimator joint
statistics, the scheduler naively trusts the channel estimates to be
actual channel states and transmits at the rate allowed in this
state. Under this scheduling structure, for the single-hop network
we consider, the stability region is given in Appendix C by
\begin{eqnarray}
\tilde{\bm \Lambda}= \sum_{\hat{\bm c}\in \mathcal{S}^N}
P_{\widehat{\bm C}}(\hat{\bm c}) \cdot \mathcal{CH}\Big[\ \bm 0, P(C_i{\geq} \hat{c}_i \big
|\widehat{C}_i{=}\hat{c}_i\ ) \hat{c}_i \cdot \vec{\bm
1}_i;i=1,\cdots, N \Big].
\end{eqnarray}
For a two-user single-hop network, this region is
plotted in Fig.~\ref{fig:example} along-side the network stability
region when full knowledge of the channel/estimator joint statistics
is available at the scheduler and hence rate adaptation is
performed. The channel between the base station and each user is
independent and binary ($\mathcal{S}=\{0.2,1 \}$) with
$P(C_i=1)=0.8, \ for \ i=1,2$. For different mismatch between the
channel and the estimate, Fig.~\ref{fig:example} plots the stability
region of the system when rate adaptation is performed and when it
is not. Note the significant reduction in the stability region when
rate adaptation is not performed. This loss increases with increase
in the degree of channel-estimator mismatch. The preceding example
underscores the importance of rate adaptation and hence the need to
learn the channel/estimator joint statistics. We now proceed to
introduce our joint statistics learning-scheduling policy. 

\subsection{Joint Statistics Learning - Scheduling Policy}

We design the policy with the following main components: (1) The
fraction of time slots the policy spends in learning the
channel/estimator joint statistics is fixed at $\gamma\in(0,1)$, (2)
The worst-case rate of convergence of the statistics learning
process is maximized.
We formally introduce the policy next, followed by a discussion on
the policy design. \\
\fbox{
\parbox[c]{0.96\linewidth}{
\vspace{11pt}\textbf{Joint statistics learning-scheduling policy \\ (parameterized by $\gamma$)}\\
\vspace{0pt}

\noindent \emph{(1)} In each slot, the scheduler first decides
whether to explore the channel of one of the users or transmit data
to one of the users. Specifically, it randomly decides to explore
the channel of user $i$ with probability $x^i_{\hat{c}_i}/N$ where
$\sum_{i=1}^Nx^i_{\hat{c}_i}/N<1$. The quantity $x^i_{\hat{c}_i}\in
(0,1]$ is a function of $\gamma$ and the channel estimate,
$\hat{c}_i$, of user $i$. It is optimized to maximize the worst-case
rate of convergence of the statistics learning mechanism subject to
the $\gamma$ constraint. We postpone the discussion on this
optimization to Proposition~\ref{prop:waterfill}. Note that, we have
dropped the
time index from the estimates for ease of notation.\\

\noindent \emph{(2)} If a user is chosen for exploration, this time
slot becomes an observing slot. Call the chosen user as $e$. The
scheduler now sends data at a rate $r$ that is chosen uniformly at
random from the set $\mathcal{S}$. Let the quantity $\xi(t)$
indicate whether the transmission was successful or not:
\begin{align}
\xi(t)= \textbf{1} ( c_e\geq r \big),\nonumber
\end{align}
}}

\fbox{
\parbox[c]{.94\linewidth}{
\vspace{5pt} where, recall, $c_e$ denotes the current channel state
of user $e$. Let $\Theta_{i,\hat{c},r}$ denote the set of
exploration time slots when the channel estimate of user $i$ was
$\hat{c}$ and user $i$ was explored with rate $r$. Thus, the current
slot is added to the set $\Theta_{e,\hat{c}_e,r}$.
Now, an estimate of the quantity $P(C_e\geq r \Big
|\widehat{C}_e=\hat{c}_e)$ is obtained using the following update:
\begin{align}
\nonumber \widehat{P}_t(C_e\geq r \Big
|\widehat{C}_e=\hat{c}_e)=\frac{\sum_{k \in \Theta_{e,\hat{c}_e,r}}
\xi(k)}{|\Theta_{e,\hat{c}_e,r}|}
\end{align}
where $|\mathcal{V}|$ denotes the cardinality of set
$\mathcal{V}$. We assume $\widehat{P}_t(C_e|\widehat{C}_e)$ to be uniform when
$\Theta_{e,\hat{c}_e,r}=\emptyset$, i.e., $\widehat{P}_t(C_e\geq r
\Big |\widehat{C}_e=\hat{c}_e)=1-r/|\mathcal{S}|$.\\

\emph{(3)} With probability
$1-\sum_{i=1}^{N}\frac{x^i_{\hat{c}_i}}{N}$, no user is chosen for
exploration and the slot is used for data transmission. The
scheduler follows policy $\Psi$ introduced in the previous section
with $P(C_i\geq r \Big |\widehat{C}_i=\hat{c}_i \ )$ replaced by the
estimate $\widehat{P}_t(C_i\geq r \Big |\widehat{C}_i=\hat{c}_i \
)$.
} }\\ \vspace{12pt}

\begin{figure*}
\centering
\includegraphics[width=5.5in]{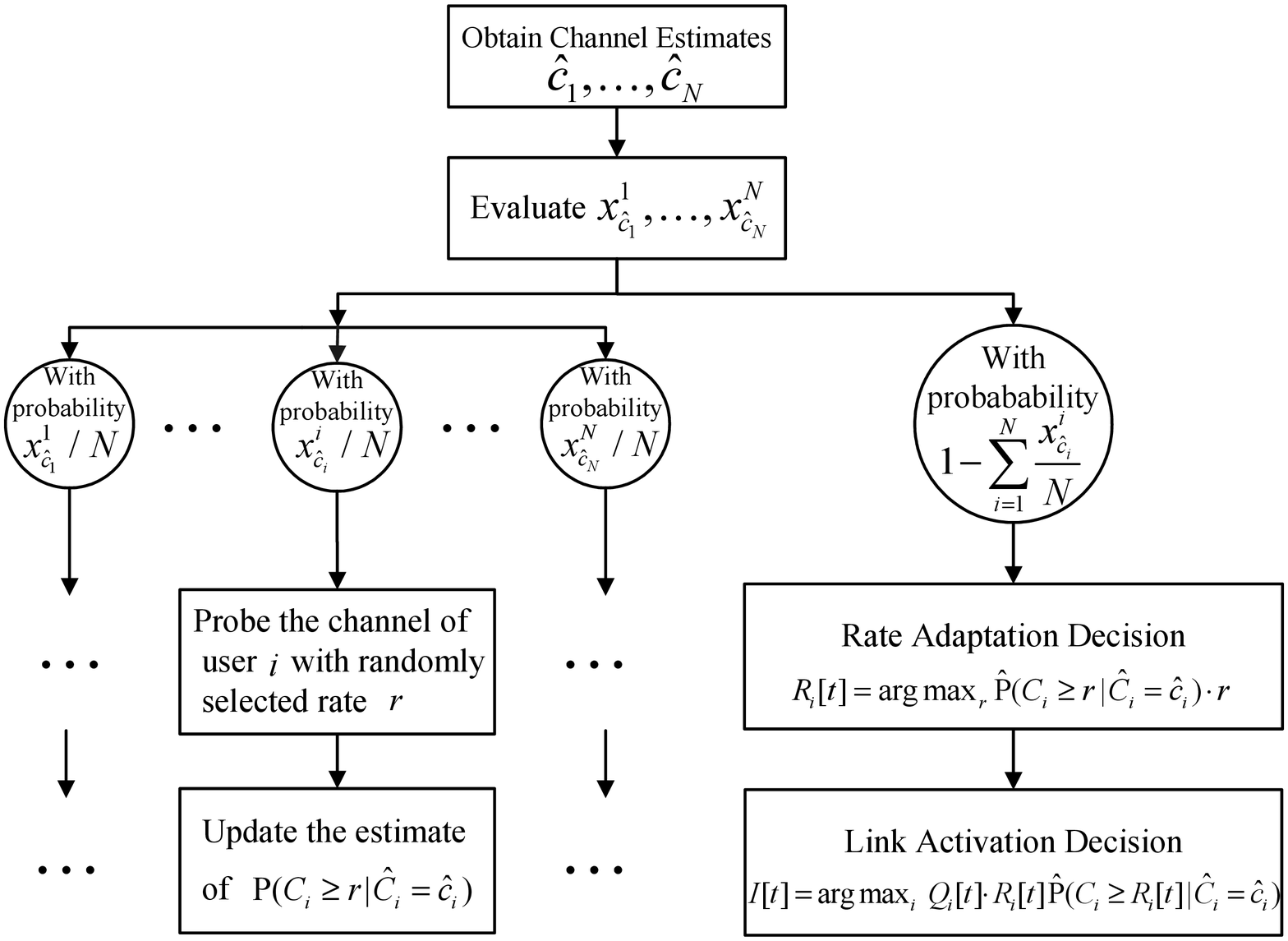}
\caption{Illustration of the joint channel learning - scheduling
policy.} \label{fig:policyillustration}
\end{figure*}
\vspace{-4pt}
An illustration of the proposed policy is provided in
Fig.~\ref{fig:policyillustration}. We now discuss the design of the
quantities $x^i_{\hat{c}_i}$, $\hat{c}_i \in \mathcal{S}$,
$i\in\{1,\ldots,N\}$. Let
$\eta_{i,\hat{c}_i}=P(\widehat{C}_i=\hat{c}_i)\frac{x^i_{\hat{c}_i}}{N}$
be a measure of how often the channel of user $i$ is explored when
the estimate is $\hat{c}_i$. For fairness considerations, we impose
the following constraint in addition to the $\gamma$-constraint
discussed earlier:
\begin{align}
\nonumber \sum_{\hat{c}_i \in \mathcal{S}}\eta_{i,\hat{c}_i}=\gamma
/ N.
\end{align}

The preceding constraint ensures that each user's channel is
explored for an equal fraction, $\gamma/N$, of the total time slots.
From strong law of large numbers, with probability one, $\widehat{P}_t(C_i\geq r \Big
|\widehat{C}_i=\hat{c}_i )$ will converge to $P(C_i\geq r \Big
|\widehat{C}_i=\hat{c}_i )$ as $t$ tends to infinity. The rate of
convergence of the channel/estimate joint statistics, parameterized
by the user and the channel estimate, is given by the following
lemma. Henceforth, we drop the suffix $i$ from $\hat{c}_i$ for
notational convenience.
%
%
\begin{lemma}
\begin{align}
\nonumber \limsup_{t\rightarrow \infty} \frac{\widehat{P}_t(C_i\geq
r \big |\widehat{C}_i=\hat{c} ){-}P(C_i\geq r \big
|\widehat{C}_i=\hat{c} )}{\sqrt{\frac{ \log \log
(\frac{\eta_{i,\hat{c}}t}{|\mathcal{S}|})}{(\frac{\eta_{i,\hat{c}}t}{|\mathcal{S}|})}}}=\sqrt{2}\sigma
\end{align}
almost surely (a.s.), where
\begin{align}
\sigma=\sqrt{P(C_i\geq r \big
|\widehat{C}_i=\hat{c})(1-P(C_i\geq r \big |\widehat{C}_i=\hat{c}
))} \nonumber.
\end{align}
\end{lemma}

\begin{proof}
We use $N_{r}[t]$ to denote the number of exploration slot
corresponding to estimated channel $\widehat{C}_i=\hat{c}$ and rate $r$. We express the left hand side of the equation
in the lemma as follows.
\begin{align}
\nonumber &\limsup_{t\rightarrow \infty} \frac{\widehat{P}_t(C_i\geq
r \Big |\widehat{C}_i=\hat{c} )-P(C_i\geq r \Big
|\widehat{C}_i=\hat{c} )}{\sqrt{(2 \log \log (\eta_{i,\hat{c}} t / |\mathcal{S}|)) /
(\eta_{i,\hat{c}} t/ |\mathcal{S}|)}} \\
=& \limsup_{t\rightarrow \infty} \frac{\widehat{P}_t(C_i\geq r \Big
|\widehat{C}_i=\hat{c} )-P(C_i\geq r \Big |\widehat{C}_i=\hat{c}
)}{\sqrt{(2 \log \log N_r[t]) / N_r[t]}} \ \cdot \nonumber \\
& \hspace{0.5in} \sqrt{\frac{\log \log
N_r[t]}{\log \log (\eta_{i,\hat{c}} t/ |\mathcal{S}|)} \cdot \frac {\eta_{i,\hat{c}} t/ |\mathcal{S}|}{N_r[t]}}
\end{align}

From Law of Iterated Logarithm (\cite{Billingsley}), we get
\begin{align}
\limsup_{t\rightarrow \infty} \frac{\widehat{P}_t(C_i\geq
r \Big |\widehat{C}_i=\hat{c} )-P(C_i\geq r \Big
|\widehat{C}_i=\hat{c})}{\sqrt{(2 \log \log N_{r}[t]) /
N_{r}[t]}}=\sigma
\end{align}
almost surely. We also have
\begin{align}
\nonumber \frac{\log \log N_r[t]}{\log \log (\eta_{i,\hat{c}} t/ |\mathcal{S}|)}=&
1+\frac{\log \log N_r[t]-\log \log (\eta_{i,\hat{c}} t/ |\mathcal{S}|)}{\log \log
(\eta_{i,\hat{c}} t/ |\mathcal{S}|)}\\
=& 1+\frac{\log (1+\frac{\log [N_r[t] / (\eta_{i,\hat{c}} t/ |\mathcal{S}|)]}{\log (\eta_{i,\hat{c}}
t/ |\mathcal{S}|)})}{\log \log (\eta_{i,\hat{c}} t/ |\mathcal{S}|)}.
\end{align}

Because $\{ N_r[t] \}$ is a renewal process (\cite{Gallager}) with inter-renewal
time $(\eta_{i,\hat{c}}/ |\mathcal{S}|)^{-1}$, we will have
\begin{align}
\lim_{t\rightarrow \infty} N_r[t] / (\eta_{i,\hat{c}} t/ |\mathcal{S}|)=1 \quad almost \ surely.  \nonumber
\end{align}

Hence (6) tends to $1$ almost surely. Substituting equation (5) and (6) into (4) we get
\begin{align}
\nonumber &\limsup_{t\rightarrow \infty} \frac{\widehat{P}_t(C_i\geq
r \Big |\widehat{C}_i=\hat{c})-P(C_i\geq r \Big
|\widehat{C}_i=\hat{c} )}{\sqrt{(2 \log \log (\eta_{i,\hat{c}} t/ |\mathcal{S}|)) /
(\eta_{i,\hat{c}} t/ |\mathcal{S}|)}}=\sigma,
\end{align}
almost surely.
\end{proof}

\vspace{8pt}

Note from the preceding lemma that, for each $\{i,\hat{c}\}$, the
higher the quantity $\eta_{i,\hat{c}}$, the faster the convergence
of $\widehat{P}_t(C_i\geq r \big |\widehat{C}_i=\hat{c})$. Also note
that, for each user $i$, the channel estimate $\hat{c}$ with the
slowest convergence affects the overall convergence performance for
that user $i$. Taking note of this, we proceed to design
$x^i_{\hat{c}}$ that maximizes the lowest convergence rate -- the
bottleneck.

The optimization problem $(U)$ for user $i$ is given by
\begin{align}
\max_{x^i_{\hat{c}}} \quad \min_{\hat{c}} \quad
&\eta_{i,\hat{c}}=\frac{1}{N}P(\widehat{C}_i=\hat{c})x^i_{\hat{c}}\nonumber \\
s.t. \hspace{0.012\textwidth} \qquad &\sum_{\hat{c} \in
\mathcal{S}} \eta_{i,\hat{c}}=\frac{\gamma}{N} \nonumber \\
&0< x^i_{\hat{c}} \leq 1,  \quad \textrm{for all}  \quad \hat{c} \in
\mathcal{S}\nonumber
\end{align}

%

For ease of exposition, we assume, without loss of generality, that
$P(\widehat{C}_i=s_1) \leq P(\widehat{C}_i=s_2) \leq \cdots \leq
P(\widehat{C}_i=s_{|\mathcal{S}|})$. Let $[ x_{s_1}^{i*},
x_{s_2}^{i*}, \cdots ,x_{s_{|\mathcal{S}|}}^{i*}]$ be the optimal
solution to the above problem.  We now record the structural
properties of the optimal solution.

%

\begin{figure*}
\centering
\includegraphics[width=3.2in]{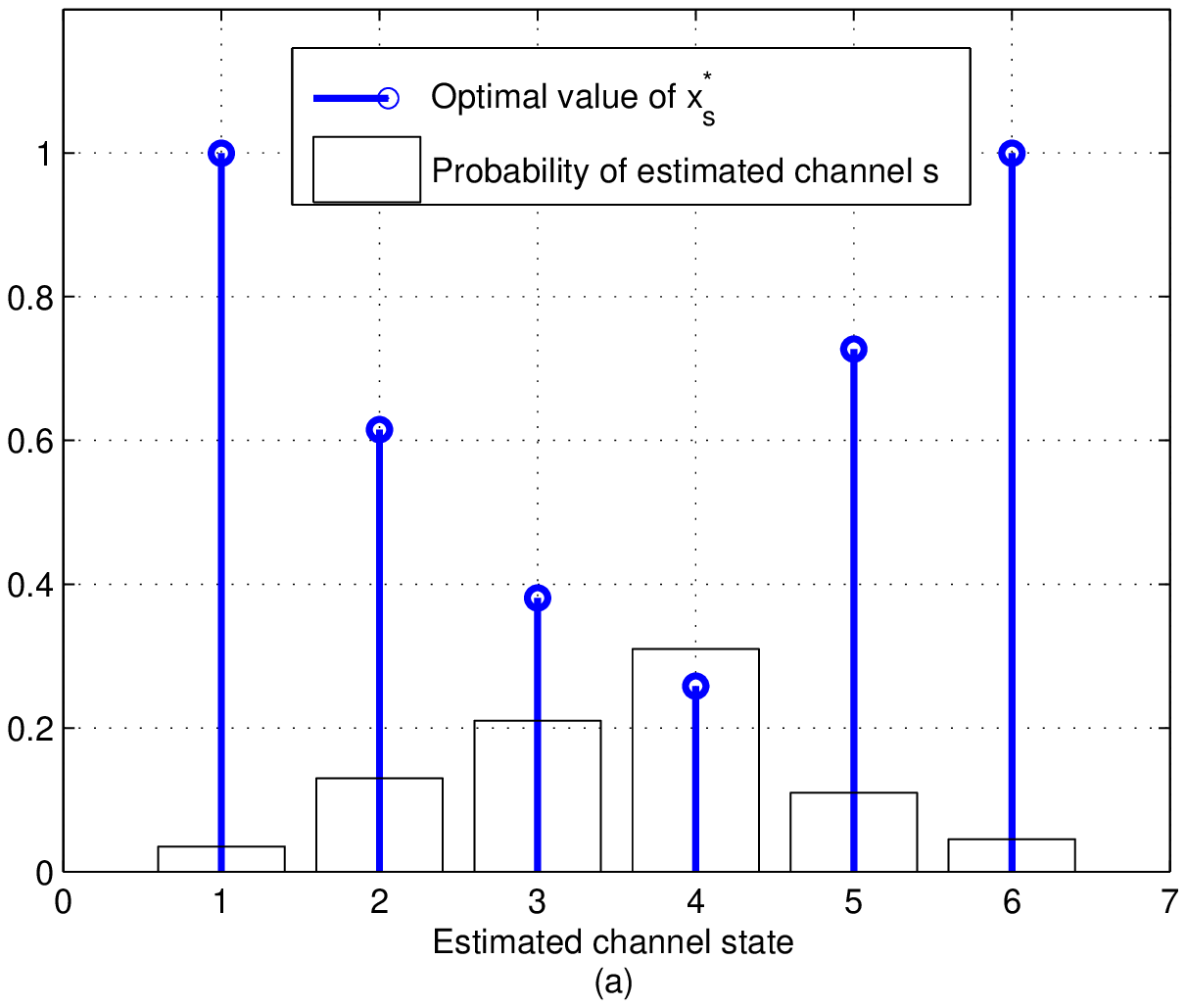}
\includegraphics[width=3.2in]{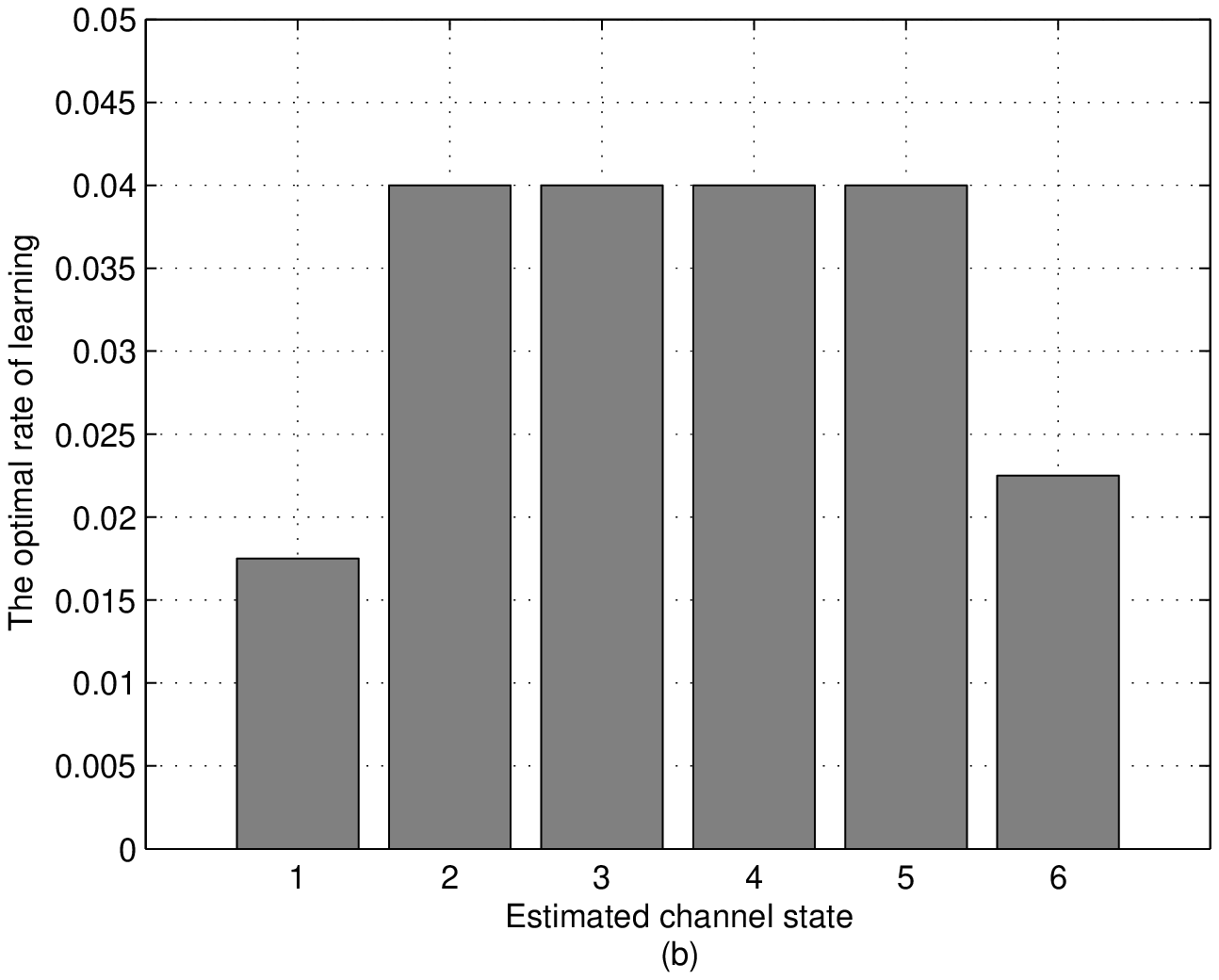}
\caption{Illustration of the design of optimal $x^*_{\hat{c}}$ when
$N=2$, $\gamma=0.2$ and $\mathcal{S}=\{1,\ldots, 6\}$.}
\label{fig:xallocation}
\end{figure*}

\begin{proposition}\label{prop:waterfill}
The solution $x_{s_k}^{i*}$, $\forall
k\in\{1,\ldots,|\mathcal{S}|\}$, to the optimization problem (U) can
be obtained with the following
algorithm:\vspace{5pt}\\
\fbox{
\parbox[c]{0.92\linewidth}{
\begin{itemize}
\item[(1)]
Initialization: Let $k=1$; $\Gamma=\emptyset$, $\omega=0$;\\
\item[(2)]
If $P(\widehat{C}=s_k) \geq \frac{\gamma-\sum_{s_j\in\Gamma}
P(\widehat{C}_i=s_j)}{|\mathcal{S}|-w}$, then,
\begin{align}
\nonumber  \forall l\geq k, \
x_{s_l}^{i*}=\frac{\gamma-\sum_{s_j\in\Gamma}
P(\widehat{C}_i=s_j)}{(|\mathcal{S}|-\omega) \cdot
P(\widehat{C}=s_l)}.
\end{align}
Algorithm terminates.\\
\item[(3)]
Otherwise $x_{s_k}^{i*}= 1$, $\Gamma=\Gamma \cup s_k$,
$\omega=\omega+1$, $k=k+1$. If $\Gamma=\mathcal{S}$, algorithm
terminates, otherwise repeat Step (2).
\end{itemize}
} }\\
\end{proposition}

\textit{Proof Outline:}
The proof proceeds by establishing two
crucial properties of the optimal solution. First, define $\Omega_i$
as the set of all channel estimates $s_k$ such that the optimal
$x_{s_k}^{i*}=1$. Thus $\Omega_i=\cup_k\{s_k: \ x_{s_k}^{i*}=1 \}$.
If no such estimate exists, $\Omega_i=\emptyset$. The optimal
solution has the following properties:
\begin{itemize}
\item[(1)] If \ $\Omega_i=\emptyset$ then $P(\widehat{C}_i=s_k)x_{s_k}^{i*}=\gamma/|\mathcal{S}|, \forall k$.
\item[(2)] If $\Omega_i \neq \emptyset$, then $x_{s_1}^{i*}=1$.
\end{itemize}

Recall that the channel states are ordered such that
$P(\widehat{C}_i=s_1) \leq P(\widehat{C}_i=s_2) \leq \cdots \leq
P(\widehat{C}_i=s_{|\mathcal{S}|})$. The first property essentially
says that if there does not exist a channel estimate $s$, for which
$x_{s}^{i*}=1$, then the optimal solution is such that the learning
rate ($\frac{P(\widehat{C}=s_k)x_{s_k}^{i*}}{N}$) is uniform
($\frac{\gamma/|\mathcal{S}|}{N}$) for all $s_k$, $k\in \{1,\ldots,
|\mathcal{S}|\}$. Because, otherwise, there is always room to
improve the bottleneck convergence rate by redesigning the
quantities $x_{s_k}^{i*}$.
The second property says that whenever there exists an estimate
$s_{k,k\neq 1}$ for which $x_{s_k}^{i*}=1$, the estimate $s_1$ acts
as a bottleneck, and the optimal value of $x_{s_1}^{i*}$ must be 1.
The proposed algorithm now checks whether a solution yielding
uniform convergence rate is feasible. If so, the solution is
trivially given by
$x_{s_k}^{i*}=\frac{1}{P(\widehat{C}_i=s_k)}\frac{\gamma}{|\mathcal{S}|}$,
for all $k\in \{1,\ldots,|\mathcal{S}|\}$. Otherwise, using the
preceding properties, the algorithm assigns $x_{s_1}^{i*}=1$ and
goes on to solve the reduced optimization problem over
$x_{s_2}^{i*}\ldots,x_{s_{|\mathcal{S}|}}^{i*}$, iteratively.
Details of the proof can be found in Appendix D.
\vspace{8pt}

The proposed algorithm is illustrated in Fig.~\ref{fig:xallocation}
when $\gamma=0.2$, $N=2$ and $\mathcal{S}=\{1,\ldots, 6\}$. Focusing
on User $1$, Fig.~\ref{fig:xallocation}(a) plots the probability of
the estimated channels and the optimal values of $x^{1*}_{s}$, $s
\in \mathcal{S}$. Note that, the lower the value of
$P(\widehat{C}_1=s)$, the higher the assigned $x^{1*}_{s}$, since
the algorithm maximizes the bottleneck convergence rate
$\frac{P(\widehat{C}=s)x_{s}^{1*}}{N}$. This is further illustrated
in Fig.~\ref{fig:xallocation}(b) where the optimized convergence
rate is shown to be `near uniform', underlining the minmax nature of
the optimization. Note that the structure of the minmax algorithm
bears some similarity with the water-filling algorithm used in power
allocation across parallel channels (\cite{Tse}). There the algorithm
tries to `equalize' the sum of two components (signal and noise
powers) across channels, while the minmax algorithm we propose tries
to `equalize' the \textit{product} of two components
($P(\widehat{C}_i=s)$ and $x_{s}^{i*}$).

We now perform a stability region analysis of the proposed policy.
Define the stability region of a policy as the exhaustive set of
arrival rates such that the network queues are rendered stable under
the policy. The stability region of the proposed policy,
parameterized by $\gamma\in(0,1)$, is recorded below.

\begin{proposition}
The stability region $\Lambda'_{\gamma}$ of the proposed policy is given by
\begin{align}
\nonumber \bm \Lambda'_{\gamma}=\{\bm \lambda \ \ s.t. \ \ \frac{\bm
\lambda}{1-\gamma} \in \bm \Lambda\}\triangleq (1-\gamma)\bm
\Lambda.
\end{align}
where $\bm \Lambda$ is the stability region of the network when
complete channel/estimator joint statistics is available at the
scheduler.
\end{proposition}

\textit{Proof Outline:} The proof proceeds by showing that, under
the proposed joint statistics learning - scheduling policy, the
instantaneous maximal sum of the queue weighted achievable rates,
with sufficient time, can be arbitrarily close to the case when
perfect knowledge of the statistics is available. Details are
provided in Appendix E.
\vspace{8pt}

\begin{figure*}
\centering
\includegraphics[width=3.4in]{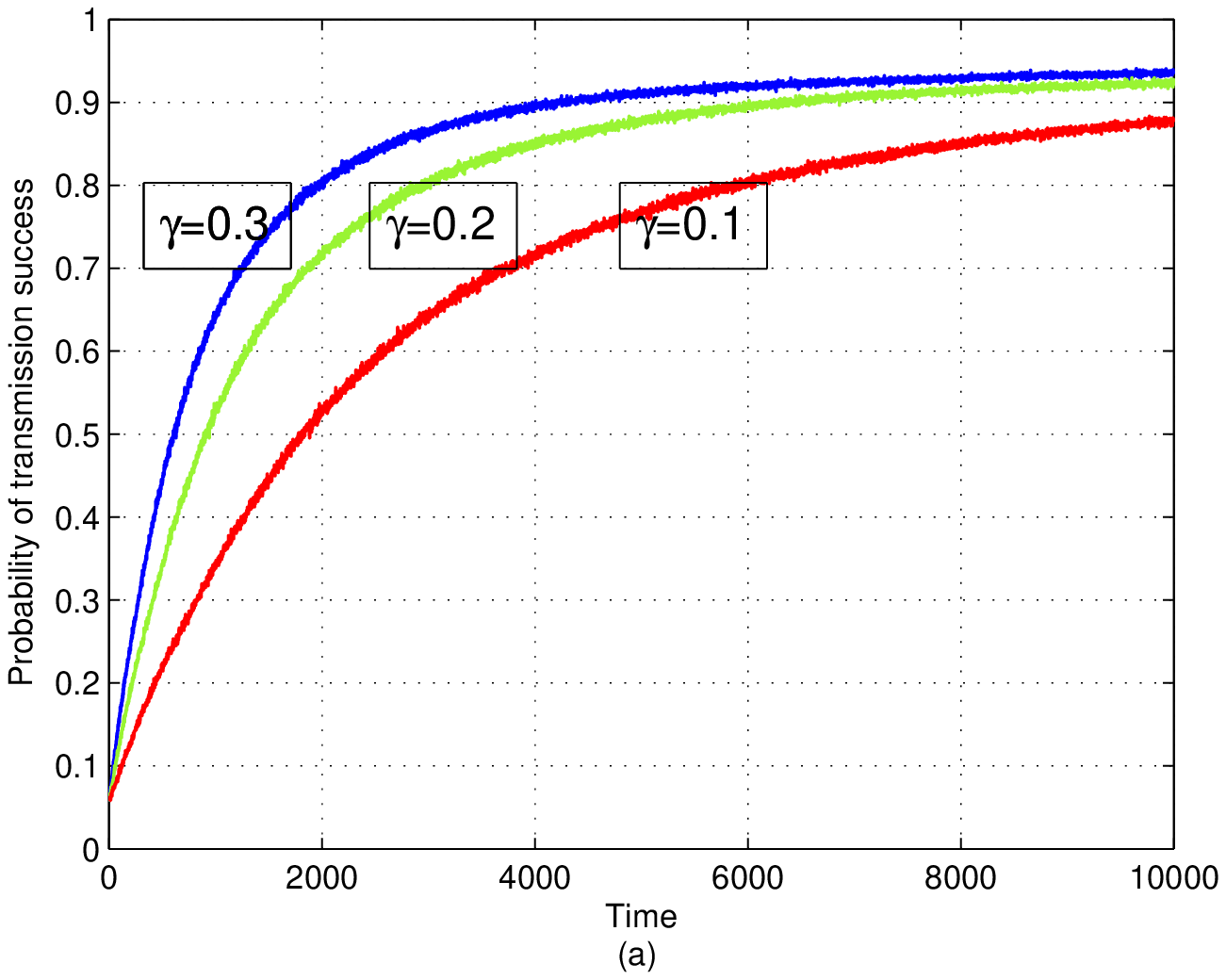}
\includegraphics[width=3.4in]{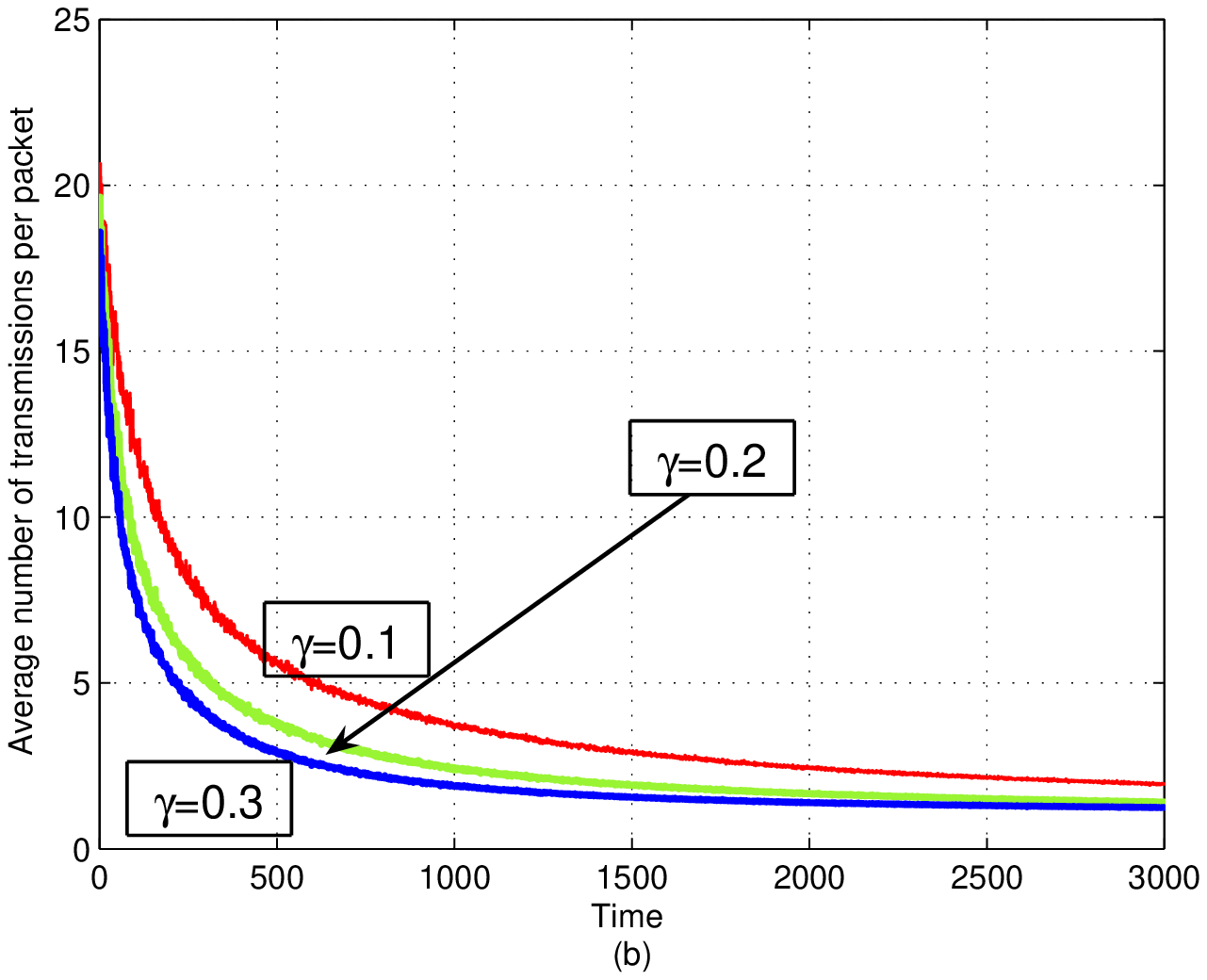}
\caption{Illustration of the time evolution of the probability of
successful packet transmission and average number of transmissions
needed per packet, for various values of $\gamma$.}
\label{fig:trsuc}
\end{figure*}

\subsection{Throughput - Delay Tradeoff}

As $\gamma\rightarrow 0$, the proposed policy has a stability region
that can be arbitrarily close to the system stability region $\bm
\Lambda$. The trade-off involved here is the speed of convergence
and hence queueing delays before convergence. Since an analytical
study of this trade-off appears complicated, we proceed to perform a
numerical study. The simulation setup is described next.

We use \textit{i.i.d.} Rayleigh fading channels with minimum mean
square error (MMSE) channel estimator as seen in \cite{Junshan1} and
\cite{Hassibi}. The channel model is given by
\begin{align}
Y=\sqrt{\rho} h X+\nu, \nonumber
\end{align}
where $X, Y$ correspond to transmitted and received signals,
$\rho$ is the average SNR at the receiver, and $\nu$ is the additive
noise. Both $h$ and $\nu$ are zero-mean complex Gaussian random
variables, i.e., with probability density $\mathcal{CN}(0,1)$. Let
$\hat{h}$ denote the estimate of the channel and $\tilde{h}$ denote
the estimation error. Under the channel statistics assumed,
$\tilde{h}$ is zero-mean complex Gaussian with variance $\beta$,
where the value of $\beta$ depends on the resources allocated for
estimation (\cite{Hassibi2}). Given the value of $h$, the channel rate
is $R=\log (1+ \rho |h|^2)$. We quantize the transmission rate to
make the channel state space to be discrete and finite. We assume a
two-user network and fix $\beta=0.1$ and $\rho=50$ for both users'
channels. We study the average behavior of the proposed policy by
implementing it over $10000$ parallel queuing systems.


We first study the time evolution of the probability of transmission
success for different values of $\gamma$. Fig.~\ref{fig:trsuc}(a)
shows that, for any $\gamma$, the probability of successful
transmission increases as the accuracy of the estimate of the
channel/estimator joint statistics improves with time. Also, as
expected, the larger the value of $\gamma$ is, the faster is the
improvement in the probability of successful transmission. Note that
higher transmission success probability essentially means lesser
number of retransmissions. This is illustrated in
Fig.~\ref{fig:trsuc}(b).

In Fig.~\ref{fig:delay}, we study the time evolution of the average
packet delay - the delay between the time a packet enters the queue
and the time it leaves the head of the queue - for various values of
$\gamma$. Note that $\gamma$ influences the average delay through
(1) the average number of retransmissions and (2) the fraction of
time slots available for transmissions. It is expected that the
nature of the influence of $\gamma$ on the average delay depends on
whether the estimate of the channel/estimator joint statistics has
reached convergence or not. After convergence, the average delay is
influenced by $\gamma$ solely through the fraction of time slots
available for transmissions. Thus, after convergence, the higher the
value of $\gamma$, the higher the average delay. This is illustrated
in Fig.~\ref{fig:delay}. \textit{Before} convergence, however, the
effect of $\gamma$ on the average delay is not straightforward.
Fig.~\ref{fig:delay}, along with the fact that higher $\gamma$
results in faster convergence, suggests the following: before
convergence, $\gamma$ influences the average delay predominantly
through the average number of retransmissions, resulting in
decreasing average delay for increasing $\gamma$. In fact,
Fig.~\ref{fig:delay} suggests the existence of a larger phenomenon:
the trade-off between throughput (the stability region) and the
delay before convergence.

\begin{figure}[!h]
\centering
\includegraphics[width=4.5in]{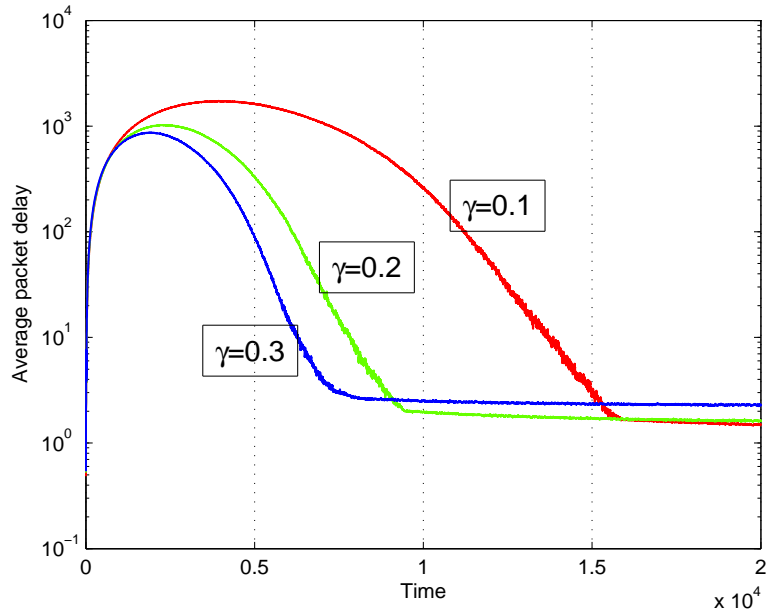}
\caption{Illustration of the average packet delay over time for
various values of $\gamma$.} \label{fig:delay}
\end{figure}

\section{Conclusion}
We studied scheduling with rate adaptation in single-hop queueing
networks, under imperfect channel state information. Under complete
knowledge of the channel/estimator joint statistics at the
scheduler, we characterized the network stability region and
proposed a maximum-weight type scheduling policy that is throughput
optimal. Under incomplete knowledge of the channel/estimator joint
statistics, we designed a joint statistics learning - scheduling
policy that maximizes the worst case rate of convergence of the
statistics learning mechanism. We showed that the proposed policy
can be tuned to achieve a stability region arbitrarily close to the
network stability region with a corresponding trade-off in the
average packet delay before convergence and the time for
convergence.

\appendix
\appendixpage

\section{Proof of Proposition 1}

\begin{proof}
\emph{(Sufficiency)} Define the Lyapunov function $L(\bm
Q[t])=\sum_{i=1}^{N} Q^2_i[t]$. Recall the queue dynamic equation
given by Equation (1), the Lyapunov drift can be written as
\begin{align}
\nonumber \Delta L[t]&=E \Big[ \sum_{i=1}^{N} Q_i^2[t+1]-Q_i^2[t]
\Big | \ \bm Q[t] \ \Big ]\\
& \leq B+ 2\sum_{i=1}^{N} Q_i[t] \Big(\lambda_i - \ E\Big[\bm 1
(I[t]=i) \cdot R[t] \cdot \bm 1(R[t] \leq C_i[t]) \Big | \ \bm Q[t] \ \Big ] \ \Big)
\end{align}
where
\begin{align}
B = \sum_{i=1}^{N} E\big[\bm 1 (I[t]=i) \cdot R^2[t] \bm 1 (R[t]
\leq C_i[t])+ A_i^2[t] \Big | \ \bm Q[t] \big].\nonumber
\end{align}

Noting that $B$ is bounded. Let $r_i^*(\hat{c}_i)=\arg\max_{r_i} P(C_i\geq
r_i \big |\widehat{C}_i=\hat{c}_i\ ) \cdot r_i$. Consider any arrival vector $\bm
\lambda$ strictly within the interior of $\bm \Lambda$. For each
channel state $\hat{\bm c}$, there exist scaling vector $\bm
\alpha^{\hat{\bm c}}$ and $\delta>0$ associated with it such that
\begin{align}
\lambda_i +\delta < \sum_{\hat{\bm c}\in \mathcal{S}^N}
\hat{\pi}_{\hat{\bm c}} \cdot \alpha_i^{\hat{\bm c}} \cdot P(C_i\geq
r_i^*(\hat{c}_i) \big |\widehat{C}_i=\hat{c}_i\ ) \cdot
r_i^*(\hat{c}_i),
\end{align}
for any user $i$,  where $\sum_{i=1}^N \alpha_i^{\hat{\bm c}}=1$ for
$\forall \hat{\bm c}$.\\

Therefore, we can design a scheduling policy that does the
following: At the channel estimation $\widehat{\bm C}[t]=\hat{\bm
c}[t]$, channel $i$ is activated with probability
$\alpha_i^{\hat{\bm c}}$. The rate
allocated to channel $i$ will be $r_i^*(\hat{c}_i[t])$.\\

Then the service rate of user $i$ will be:
\begin{align}
\mu_i =E \big[\bm 1(I[t]=i) \cdot R_i[t] \cdot \bm 1(R_i[t]\leq
C_i[t]) \big] =\sum_{\hat{\bm c}\in \mathcal{S}^N} \pi_{\hat{\bm c}}
\cdot \alpha_i^{\hat{\bm c}} \cdot P(C_i\geq r_i^*(\hat{c}_i)  \big
|\widehat{C}_i=\hat{c}_i\ ) \cdot r_i^*(\hat{c}_i)
\end{align}

Substitute (9) to (8) we have
\begin{align}
\lambda_i-\mu_i < -\delta.
\end{align}

Noting that in this policy the rate adaptation and link activation is completely determined by the channel estimation, and does not
rely on queue length information, and therefore in this case
\begin{align}
E\Big[\bm 1 (I[t]=i) \cdot R[t] \cdot \bm 1(R[t] \leq
C_i[t]) \Big | \ \bm Q[t] \ \Big ] \ \Big) = E\Big[\bm 1 (I[t]=i)
\cdot R[t] \cdot \bm 1(R[t] \leq C_i[t]) \ \Big ] \ \Big) =  \mu_i
\end{align}

Substitute (10) (11) into (7), the Lyapunov drift function now becomes
\begin{align}
\nonumber \Delta L[t]\leq B- 2\delta \sum_{i=1}^{N} Q_i[t].
\end{align}

Because the scheduling and rate adaptation decision only depends on the
current queue length and current channel estimate state, the queue
evolves as a Markov Chain. According to Foster-Lyapunov Stability
criterion \cite{Sean}, the queues will be stable.\\

\emph{(Necessity)} From strict separation theorem \cite{Separation}, for any arrival
rate vector out side the proposed region $\bm \Lambda$, there exist
$\bm \beta$, $\delta
>0$, such that for any vector $\vec{\bm \nu}$ inside the stability region
$\bm \Lambda$
\begin{align}
\sum_{i=1}^{N} \beta_i (\lambda_i - \nu_i) \geq \delta \nonumber
\end{align}

Define Lyapunov function $L(\bm Q[t])=\sum_{i=1}^{N} \beta_i
Q_i[t]$. For any stationary scheduling policy that makes $I[t]$ and $R[t]$ decisions, we will have the following Lyapunov drift expression
\begin{align}
E\Big [L(\bm Q[t+1])- L(\bm Q[t])\Big | \bm Q[t] \Big]=&
\sum_{i=1}^N \beta_i E\Big [A_i[t]- \bm 1(I[t]=i) \cdot \bm
1(R_i[t] \leq C_i[t] ) \cdot
R_i[t] \Big | \bm Q[t] \Big]\nonumber \\
= & \sum_{i=1}^N \beta_i \Big[\lambda_i- E \big [\bm 1(I[t]=i)
\cdot \bm 1(R_i[t] \leq C_i[t] ) \cdot R_i[t] \Big |  \bm Q[t]\big]
\Big]
\end{align}

Let $\mu_i=E \big [\bm 1(I[t]=i) \cdot \bm 1(R_i[t] \leq C_i[t] )
\cdot R_i[t] \Big | \bm Q[t]\big] \Big]$.

Next we are going to show that $\vec{\bm \mu} \in \bm \Lambda$.
Consider
\begin{align}
&E \Big [\bm 1(I[t]=i) \cdot \bm 1(R_i[t] \leq C_i[t] ) \cdot R_i[t]
\Big |  \bm Q[t] \big] \Big] \nonumber \\
=& E \Big [ E \big [\bm 1(I[t]=i) \cdot \bm 1(R_i[t] \leq C_i[t] )
\cdot R_i[t]
\Big |  \bm Q[t]; \widehat{\bm C}[t] \big] \Big] \nonumber \\
= &\mathop{ \sum}_{\hat{\bm c_t} \in \mathcal{S}^{N}} \pi_{\hat{\bm
c}_t} \cdot E \Big [\bm 1(I[t]=i) \cdot \bm 1(R_i[t] \leq C_i[t] )
\cdot R_i[t] \Big |\bm Q[t]; \widehat{\bm C}[t]=\hat{\bm c}_t
\Big] \nonumber\\
=& \sum_{\hat{\bm c} \in \mathcal{S}^{N}} \pi_{\hat{\bm c}} \cdot
\bm 1(I[t]=i) \cdot R_i[t] \cdot Pr \big ( R_i[t] \leq C_i[t] \Big |
\widehat{C}_i[t]=\hat{c} \big) \nonumber\\
\leq & \sum_{\hat{\bm c}\in \mathcal{S}^N} \pi_{\hat{\bm c}} \cdot
\bm 1(I[t]=i) \cdot r_i^*(\hat{c}_i[t]) \cdot Pr \big (
r_i^*(\hat{c}_i[t]) \leq C_i[t] \Big | \widehat{C}_i[t]=\hat{c}_i[t]
\big). \nonumber
\end{align}
The third equality holds because $R[t]$ and $I[t]$ is determined by
the current channel estimation and queue length information within the class G of stationary policies, and
also the \emph{i.i.d.} channel assumption. The above expression indicates that
$\vec{\bm \mu} \in \bm \Lambda$. Hence from (12)
\begin{align} E\Big [L(\bm Q[t+1])-
L(\bm Q[t])\Big | \bm Q[t]\Big] &= \sum_{i=1}^N \beta_i
\Big[\lambda_i- E \big [\bm 1(I[t]=i) \cdot \bm 1(R_i[t] \leq C_i[t]
) \cdot R_i[t] \Big | \bm
Q[t] \big] \Big] \nonumber \\
& \geq \delta \nonumber
\end{align}

The Lyapunov function will always have a positive drift and
therefore, the queue is unstable.
\end{proof}

\section{Proof of Proposition 2}
\begin{proof}
Assume the arrival rate vector $\bm {\lambda}$ is strictly within the
interior of stability region, there exists $\varepsilon>0$ such that
$\bm {\lambda}+\varepsilon \vec{\bm 1} \in int (\bm \Lambda)$.
Because $\bm {\lambda}$ is strictly within the stability region,
similar to the proof of Proposition 1, there exists some randomized scheduling
policy $G_0$ that stably supports the arrival rate vector $\bm
{\lambda}+\varepsilon \vec{\bm 1}$, and that $G_0$ will only depends
on the estimated channel state.

Suppose the proposed scheduling policy $\Psi$ will result in rate
allocation $R_i[t]$ and scheduling decision $I[t]$ at time $t$.
Consider the policy $G_0$ that act at the same time $t$ with the
same channel state estimate and queue lengths knowledge, we denote
its rate allocation to be $\widetilde{R}_i[t]$ and link activation
decision to be $\widetilde{I}[t]$, therefore we have
\begin{align}
& \sum_{i=1}^{N} Q_i[t] \cdot E\Big[\bm 1 (\widetilde{I}[t]=i) \cdot
\widetilde{R}[t] \cdot \bm 1(\widetilde{R}[t] \leq C_i[t]) \ \Big ] \nonumber \\
= &\sum_{i=1}^{N} Q_i[t] \cdot E\Big[ E\Big[\bm 1
(\widetilde{I}[t]=i) \cdot \widetilde{R}[t] \cdot \bm
1(\widetilde{R}[t] \leq C_i[t])\Big | \ \bm Q[t], \
\widehat{\bm C}[t] \Big ] \ \Big ] \nonumber \\
\leq & \sum_{i=1}^{N} Q_i[t] \cdot E\Big[ E\Big[\bm 1 (I[t]=i) \cdot
R[t] \cdot \bm 1(R[t] \leq C_i[t])\Big | \ \bm Q[t], \ \widehat{\bm
C}[t] \Big ]  \ \Big ]. \nonumber
\end{align}
The last inequality holds because policy $\Psi$ maximizes the left
hand side of the above inequality at every time slot. Also because
the queue of each user is stable under policy $G_0$, we have
\begin{align}
\lambda_i+\varepsilon & \leq E\Big[\bm 1 (\widetilde{I}[t]=i) \cdot
\widetilde{R}_{i}[t] \cdot \bm 1(\widetilde{R}_{i}[t]\leq
C_i[t])\Big] \nonumber
\end{align}

And therefore
\begin{align}
\sum_{i=1}^{N} Q_i[t] \cdot (\lambda_i+\varepsilon) & \leq
\sum_{i=1}^{N} Q_i[t] \cdot E\Big[\bm 1 (\widetilde{I}[t]=i) \cdot
\widetilde{R}_{i}[t] \cdot \bm
1(\widetilde{R}_{i}[t]\leq C_i[t])\Big]\nonumber \\
&\leq \sum_{i=1}^{N} Q_i[t] \cdot E\Big[ E\Big[\bm 1 (I[t]=i) \cdot
R[t] \cdot \bm 1(R[t] \leq C_i[t]) \Big | \ \bm Q[t], \ \widehat{\bm
C}[t] \Big ] \ \Big ] \nonumber
\end{align}

Substitute it back to Lyapunov drift expression (7), then we will have:
\begin{align}
\nonumber \Delta L[t]&\leq B+ 2\sum_{i=1}^{N} Q_i[t] \Big(\lambda_i
-E\Big[ E\Big[\bm 1 (I[t]=i) \cdot R[t] \cdot \bm 1(R[t] \leq
C_i[t]) \Big | \ \bm Q[t], \ \widehat{\bm C}[t] \Big ] \ \Big ]
\Big)\\
& \leq B- 2 \varepsilon \sum_{i=1}^{N} Q_i[t] \nonumber
\end{align}

Because scheduling policy $\Psi$ only depends on current queue
length and channel estimate, and because the channel process is a
\emph{i.i.d.} across time, the queue evolution under policy $\Psi$ will be a Markov
Chain. From Foster-Lyapunov criterion, the statement is proven.
\end{proof}

\section{Proof of Stability Region Without Rate Adaptation}
\begin{proof}
The proof of the statement is somehow similar to proof of proposition 1.\\

\emph{(Sufficiency)} Define the Lyapunov function $L(\bm
Q[t])=\sum_{i=1}^{N} Q^2_i[t]$, the Lyapunov drift can be written as
Equation (7). For any arrival vector $\bm \lambda$ strictly within
the interior of $\tilde{\bm \Lambda}$ and each channel state
$\hat{\bm c}$, the vector $\bm \alpha^{\hat{\bm c}}$ and $\delta>0$
satisfies
\begin{align}
\lambda_i +\delta < \sum_{\hat{\bm c}\in \mathcal{S}^N}
\hat{\pi}_{\hat{\bm c}} \cdot \alpha_i^{\hat{\bm c}} \cdot P(C_i\geq
\hat{c}_i \big |\widehat{C}_i=\hat{c}_i\ ) \cdot \hat{c}_i,
\nonumber
\end{align}
for any user $i$,  where $\sum_{i=1}^N \alpha_i^{\hat{\bm c}}=1$ for
$\forall \hat{\bm c}$.\\

The rest of the proof follows similar as in Proposition 1.

\emph{(Necessity)} Similar to the proof of Proposition 1, for any
arrival rate vector out side  $\tilde{\bm \Lambda}$, there exist
$\bm \beta$, $\delta
>0$, such that for any vector $\vec{\bm \nu}$ inside the stability region
$\tilde{\bm \Lambda}$, we have $\sum_{i=1}^{N} \beta_i (\lambda_i -
\nu_i) \geq \delta$.

Define Lyapunov function $L(\bm Q[t])=\sum_{i=1}^{N} \beta_i
Q_i[t]$, for any stationary policy that makes scheduling decision $I[t]$ and rate adaptation decision $R[t]$, again we will have the similar Lyapunov drift expression as
in Equation (12),
\begin{align}
E\Big [L(\bm Q[t+1])- L(\bm Q[t])\Big | \bm Q[t] \Big] =
\sum_{i=1}^N \beta_i \cdot \Big[\lambda_i- E \big [\bm 1(I[t]=i)
\cdot \bm 1(\widehat{C}_i[t] \leq C_i[t] ) \cdot \widehat{C}_i[t]
\Big |  \bm Q[t]\big] \Big] \nonumber
\end{align}

Let $\mu_i=E \big [\bm 1(I[t]=i) \cdot \bm 1(\widehat{C}_i[t] \leq
C_i[t] ) \cdot \widehat{C}_i[t] \Big | \bm Q[t]\big] \Big]$. We can
show $\vec{\bm \mu} \in \tilde{\bm \Lambda}$ from
\begin{align}
E \Big [\bm 1(I[t]=i) \cdot \bm 1(\widehat{C}_i[t] \leq C_i[t] )
\cdot \widehat{C}_i \Big |  \bm Q[t] \big] \Big] =& E \Big [ E \big
[\bm 1(I[t]=i) \cdot \bm 1(\widehat{C}_i[t] \leq C_i[t] ) \cdot
\widehat{C}_i[t]
\Big |  \bm Q[t]; \widehat{\bm C}[t] \big] \Big] \nonumber \\
=& \sum_{\hat{\bm c} \in \mathcal{S}^{N}} \pi_{\hat{\bm c}}
\cdot \bm 1(I[t]=i) \cdot \hat{c}_i \cdot Pr \big (\hat{c}_i \leq C_i[t] \Big |
\widehat{C}_i[t]=\hat{c}_i \big) \nonumber\\
= & \sum_{\hat{\bm c}\in \mathcal{S}^N} \pi_{\hat{\bm c}} \cdot
\bm 1(I[t]=i) \cdot \hat{c}_i \cdot Pr \big (
\hat{c}_i \leq C_i[t] \Big | \widehat{C}_i[t]=\hat{c}_i
\big). \nonumber
\end{align}

The rest follows similarly as in proof of Proposition 1.
\end{proof}

\section{Proof of Proposition 4}

For notational convenience, we drop user index $i$ in the proof.
The optimization problem (U) can be re-written as
\begin{align}
\nonumber \max_{x_{\hat{c}}} \quad \min_{\hat{c}} \quad &Pr(\widehat{C}=\hat{c})x_{\hat{c}} \\\
\nonumber s.t. \hspace{0.02\textwidth} \qquad &\sum_{\hat{c} \in
\mathcal{S}} Pr(\widehat{C}=\hat{c})x_{\hat{c}}=\gamma \\
 &\sum_{\hat{c}\in \mathcal{S}}
Pr(\widehat{C}=\hat{c})=1\nonumber  \\
\nonumber &0< x_{\hat{c}} \leq 1.
\end{align}

This problem can be transformed into a Linear Programming problem as
the following.
\begin{align}
\nonumber \max \quad & t  \\
\nonumber s.t. \quad &\Pr(\widehat{C}=\hat{c})x_{\hat{c}} \geq t\\
&\sum_{\hat{c} \in
\mathcal{S}} Pr(\widehat{C}=\hat{c})x_{\hat{c}}=\gamma \nonumber \\
&\sum_{\hat{c}\in \mathcal{S}}
Pr(\widehat{C}=\hat{c})=1 \nonumber \\
\nonumber &0 <x_{\hat{c}} \leq 1. \nonumber
\end{align}
Hence the problem has become a convex optimization
problem. We let $[t^*, x_{s_1}^*, x_{s_2}^*, \cdots ,x_{s_{|\mathcal{S}|}}^*]$
be the optimal solution to the above problem and let $\Omega=\{ s_k: \
x_{s_k}^*=1 \}$.

\begin{lemma}
The optimal solution to the optimization problem (U) must satisfy the following structural properties:

(i)\ \ If \ $\Omega=\emptyset$ then $Pr(\widehat{C}=s_k)x_{s_k}^*=\gamma/|\mathcal{S}|$ for all $k$.

(ii) Conversely, if $Pr(\widehat{C}=s_k)x_{s_k}^*=\frac{\gamma}{|\mathcal{S}|}$ for all
$k$, then $\Omega=\emptyset$, except for when $Pr(\widehat{C}=s_1)=\frac{\gamma}{|\mathcal{S}|}$.
\end{lemma}

\begin{proof}
(i). Suppose $\Omega=\emptyset$ but we don't have
$Pr(\widehat{C}=s_k)x_{s_k}^*=\frac{\gamma}{s}$ for all $k$. Let $n$ be such that
$t^*=Pr(\widehat{C}=s_n)x_{s_n}^*$. Because
$Pr(\widehat{C}=s_k)x_{s_k}^*$ values are not equal for every $k$, there exists
$m\neq n$ such that $t^*=Pr(\widehat{C}=s_n)x_{s_n}^* <
Pr(\widehat{C}=s_m)x_{s_m}^*$.

Because $\Omega=\emptyset$, $x_{s_n}^*\neq 1$ and $x_{s_m}^*\neq 1$. Let
$\Pi=\{x_{s_k}^*: t^*=Pr(\widehat{C}=s_k)x_{s_k}^* \}$ and let
$|\Pi|=\pi$. If we set
\begin{align}
x_{s_l}^+&=x_{s_l}^*+\frac{\delta}{\pi
Pr(\widehat{C}=s_l)}, \ \forall l \in \Pi \nonumber \\
x_{s_m}^-&=x_{s_m}^*-\frac{\delta}{Pr(\widehat{C}=s_m)} \nonumber
\end{align}
where $\delta>0$ is small that will guarantee that $x_{s_m}^-$ stays
positive.

We can check that in this case, still $\mathop{\sum}_{\hat{c}_i\neq
m \atop \hat{c}_i \not \in
\Pi}Pr(\widehat{C}{=}\hat{c_i})x_{s_k}^*{+}\sum_{\hat{c}_l \in
\Pi}Pr(\widehat{C}{=}\hat{c_l})x_{s_l}^+ {+} Pr(\widehat{C}=\hat{c_m})
x_{s_m}^- {=}\gamma$.

But in this case the new value of the objective function $t_{new}^*>t^*$, contradicts
to the assumption that $t^*$ is the optimal value. Therefore we must have
$Pr(\widehat{C}=s_k)x_{s_k}^*=\frac{\gamma}{|\mathcal{S}|}, \forall
i$, establishing the proof of (i).\\

(ii). When
$Pr(\widehat{C}=s_k)x_{s_k}^*=\frac{\gamma}{|\mathcal{S}|}, \forall
k$, and $x_{s_k}^*< 1$ for all $k$,  we will have $\Omega=\emptyset$.

If $\exists h, Pr(\widehat{C}=s_h)=\gamma/|\mathcal{S}|$, hence $x_{s_h}^*=1$. By assumption
\begin{align}
Pr(\widehat{C}=s_1)x_{s_1}^*=Pr(\widehat{C}=s_h)x_{s_h}^*=Pr(\widehat{C}=s_h). \nonumber
\end{align}

Because $Pr(\widehat{C}=s_1)\leq Pr(\widehat{C}=s_h)$ and $x_{s_1}^*\leq 1$, we must have $Pr(\widehat{C}=s_1)= Pr(\widehat{C}=s_h)$ and $x_{s_1}^*=1$, establishing the proof of (ii).
\end{proof}

\begin{lemma}
If $\Omega \neq \emptyset$, then $x_{s_1}^*=1$ and
$t^*=Pr(\widehat{C}=s_1)$.
\end{lemma}

\begin{proof}
We proof this Lemma by contradiction. Suppose that $x_{s_1}^*<1$.

(Case 1). If $t^*=Pr(\widehat{C}=s_1) x_{s_1}^*$, without loss of
generality, suppose $s_1$ is the only channel state that
results in $Pr(\widehat{C}=s_1) x_{s_1}^*=t^*$. Because $\Omega \neq
\emptyset$, suppose $x_{s_j}^*=1$ for some state $s_j$. If we set
$x_{s_j}^{new}=1-\delta/Pr(\widehat{C}=s_j)$ and
$x_{s_1}^{new}=x_{s_1}^*+\delta/Pr(\widehat{C}=s_1)$, where $\delta$
is small enough, we will get an new value of the objective function strictly larger than
$t^*$ while still satisfy the constraints in (U), which contradicts to the optimality of $[t^*, x_{s_1}^*, x_{s_2}^*, \cdots ,x_{s_\mathcal{S}}^*]$.

(Case 2). If $t^*<Pr(\widehat{C}=s_1) x_{s_1}^*$, suppose
$t^*=Pr(\widehat{C}=s_m) x_{s_m}^*$ for some $s_m$ and assume, with no loss of generality,
it is the only state of this kind. Because $Pr(\widehat{C}=s_m)
x_{s_m}^*<Pr(\widehat{C}=s_1) x_{s_1}^*$ and
$Pr(\widehat{C}=s_m)\geq Pr(\widehat{C}=s_1)$, we have
$x_{s_m}^*<x_{s_1}^*<1$. We can set
$x_{s_m}^{new}=x_{s_m}^*+\delta/Pr(\widehat{C}=s_m)$ and
$x_{s_1}^{new}=x_{s_1}^*-\delta/Pr(\widehat{C}=s_1)$ for $\delta$
small. Again, this change of variables will result in an new objective function value strictly larger than $t^*$, contradicts to the optimality of $t^*$.

Therefore, we have $x_{s_1}^*=1$.

Suppose we have $t^*<Pr(\widehat{C}=s_1) x_{s_1}^*$. Similar to case
2, suppose $t^*=Pr(\widehat{C}=s_m) x_{s_m}^*$ for some $s_m$, and assume such $s_m$
is unique, we will have $x_{s_m}^*<1$ because $Pr(\widehat{C}=s_m)\geq Pr(\widehat{C}=s_1)$. By letting
$x_{s_m}^{new}=1+\delta/Pr(\widehat{C}=s_m)$ and
$x_{s_1}^{new}=1-\delta/Pr(\widehat{C}=s_1)$ with $\delta$ small, we
can get a strictly larger objective function value, contradicting to the optimality of $t^*$. Therefore we must have
$t^*=Pr(\widehat{C}=s_1) x_{s_1}^*$.
\end{proof}

After we have established the above lemmas, we proceed to the proof of Proposition 4.\\

\noindent \textbf{(Proof of Proposition 4)}

\begin{proof}
\quad \\
(Case 1). First consider the case when $Pr(\widehat{C}=s_1) > \frac{\gamma}{|\mathcal{S}|}$. If in this case $\Omega \neq \emptyset$, then
from Lemma 7,
$t^*=Pr(\widehat{C}=s_1)>\frac{\gamma}{|\mathcal{S}|}$. Then
$\sum_{k=1}^{|\mathcal{S}|} Pr(\widehat{C}=s_k)x_{s_k}^* \geq
|\mathcal{S}| \cdot t^* > |\mathcal{S}| \cdot
\frac{\gamma}{|\mathcal{S}|}=\gamma$.  Therefore contradict to the constraint
$\sum_{k=1}^{|\mathcal{S}|} Pr(\widehat{C}=s_k)x_{s_k}^*=\gamma$. So
$\Omega = \emptyset$, from Lemma 6, we have
$Pr(\widehat{C}=s_k)x_{s_k}^*=\frac{\gamma}{|\mathcal{S}|}$ for all
$i$ and
\begin{align}
x_{s_k}^*=\frac{\gamma}{|\mathcal{S}| \cdot Pr(\widehat{C}=s_k)}. \nonumber
\nonumber
\end{align}

\noindent (Case 2). Consider when $\frac{\gamma}{|\mathcal{S}|} =
Pr(\widehat{C}=s_1)$. If $\Omega=\emptyset$, then from Lemma 6 (i),
$x_{s_1}^*=1$, contradict to $\Omega=\emptyset$. So $\Omega \neq \emptyset$ and from Lemma
7, $t^*=Pr(\widehat{C}=s_1)=\frac{\gamma}{|\mathcal{S}|}$. Because
\begin{align}
\sum_{i=1}^{|\mathcal{S}|} Pr(\widehat{C}=s_k)x_{s_k}^* \geq \sum_{i=1}^{|\mathcal{S}|}
t^* = |\mathcal{S}| \cdot \frac{\gamma}{|\mathcal{S}|}=\gamma, \nonumber
\end{align}
we must have $t^*=Pr(\widehat{C}=s_k)x_{s_k}^*$ for all $k$ in order to satisfy the constraint of $\sum_{k=1}^{|\mathcal{S}|} Pr(\widehat{C}=s_k)x_{s_k}^*=\gamma$.
Therefore we still have
\begin{align}
x_{s_k}^*=\frac{\gamma}{|\mathcal{S}| \cdot Pr(\widehat{C}=s_k)}.
\nonumber
\end{align}

It is easy to check here that $0<x_{s_k}^*\leq 1$ in the case 1 and case 2. We hence have justified the step (2) in Proposition 4 when $k=1$.\\

\noindent (Case 3). If we have
$Pr(\widehat{C}=s_1)<\frac{\gamma}{|\mathcal{S}|}$, then we can not
set $Pr(\widehat{C}=s_k)x_{s_k}^*=\frac{\gamma}{|\mathcal{S}|}$ for
all $k$ because $x_{s_1}^* \leq 1$. So from Lemma 6, $\Omega \neq
\emptyset$. From Lemma 7, $x_{s_1}^*=1$ and
$t^*=Pr(\widehat{C}=s_1)$.

Since now we have identified the optimal value $t^*$ of the objective function and $x_{s_1}^*$, we still need to identify the rest of the solution of $x_{s_j}^*$ for $j \neq 1$. Admitting there might be multiple solutions for those $x_{s_j}^*$, we consider the following relaxed optimization problem ($U^+$)
\begin{align}
\nonumber \max \quad & t\\
\nonumber s.t. \quad &\Pr(\widehat{C}=s_k)x_{s_k} \geq t; \ k \neq 1\\
\nonumber &\sum_{k\neq 1} Pr(\widehat{C}=s_k)x_{s_k}=\gamma^+\\
\nonumber & x_{s_1}=1; \ 0 < x_{s_k} \leq 1, k\neq 1. \nonumber
\end{align}
where $\gamma^+=\gamma-Pr(\widehat{C}=s_1)$.

It can be readily verified Lemma 7 and Lemma 8 holds for the above optimization problem with $\gamma$ and $\mathcal{S}$ substituted by $\gamma^+$ and $\mathcal{S} \backslash s_1$, respectively. Let $\hat{x}_{s_j}^*$, $j=1, \cdots, N$ be the optimal solution for the above optimization problem ($U^+$). We proceed to show $\hat{x}_{s_j}^*$ is also optimal solution to optimization problem ($U$), i.e., satisfying all the constraints of ($U$) and will preserve the optimality of $t^*$ of ($U$) identified earlier.

Let $\hat{t}^*$ be the optimal objective function to the problem ($U^+$). To show that optimal solution to ($U^+$) (i.e., $\hat{x}_{s_j}^*$ for $j \neq 1$) preserve the optimality of $t^*$, we must check that $\hat{t}^*\geq t^*$. This is indeed the case and is explained as follows. Let $\widehat{\Omega}$ be the set of estimated channel states $s_h$ such that $\hat{x}_{s_h}^*=1$ in ($U^+$). When $\widehat{\Omega}\neq \emptyset$, from Lemma 7, $\hat{t}^*=Pr(\widehat{C}=s_2)x_{s_2}^*=Pr(\widehat{C}=s_2)\cdot
1  \geq Pr(\widehat{C}=s_1) \cdot 1 \geq Pr(\widehat{C}=s_1) \cdot x_{s_1}^* \geq t^*$. When $\widehat{\Omega}= \emptyset$, from Lemma 6 we have $\forall j \neq 1$,

\begin{align}
\tilde{t}^*=\hat{x}_{s_j}^* Pr(\widehat{C}=s_j)=\frac{\gamma-Pr(\widehat{C}=s_1)}{|\mathcal{S}|-1}> \frac{|\mathcal{S}|
Pr(\widehat{C}=s_1)- Pr(\widehat{C}=s_1)}{
(|\mathcal{S}|-1)}=Pr(\widehat{C}=s_1)=t^* \nonumber
\end{align}
where the inequality is from $\frac{\gamma}{|\mathcal{S}|} > Pr(\widehat{C}=s_1)$ assumed in the beginning of Case 3. So the optimality of $t^*$ is preserved. It is also clear that the constraint $Pr(\widehat{C}=s_1) x_{s_1}^* + \sum_{j=2}^{|\mathcal{S}|} Pr(\widehat{C}=s_j) \hat{x}_{s_j}^* =\gamma$ is satisfied.

We hence face a reduced optimization problem ($U^+$), for which the optimal solution will also be optimal for the original optimization problem (U). Problem ($U^+$) takes the same form of (U) with $\gamma$ and and $\mathcal{S}$ substituted by $\gamma^+$ and $\mathcal{S} \backslash s_1$, respectively. The proposed algorithm solves the reduced optimization problem by conditioning on the reduced settings of (Case 1)-(Case 3). Hence similar proof as in (Case 1)-(Case 3) is also applicable for the iterative algorithm. By doing the proof iteratively, the optimality of the algorithm is proved.
\end{proof}

\section{Proof of Proposition 5}

We need the following lemmas before proving Proposition 4:

\begin{lemma}
For any user $i$ and any observed channel state $\hat{c}_k$, there
exist a constant $T$, such that for any time $t>T$, and any channel
states $r_1$ and $r_2$, if $P(C_i\geq r_1 \big
|\widehat{C}_i=\hat{c}_k )r_1
> P(C_i\geq r_2 \big |\widehat{C}_i=\hat{c}_k )r_2$, then
\begin{align}
\widehat{P}_t(C_i\geq r_1 \big |\widehat{C}_i=\hat{c}_k )r_1
>\widehat{P}_t(C_i\geq r_2 \big |\widehat{C}_i=\hat{c}_k )r_2 \quad almost \ surely. \nonumber
\end{align}

\end{lemma}

\begin{proof} Let $\Delta=P(C_i\geq r_1 \big |\widehat{C}_i=\hat{c}_k )r_1 -
P(C_i\geq r_2 \big |\widehat{C}_i=\hat{c}_k )r_2$. Then there
exist $\delta>0$ and $\epsilon>0$ such that $\delta r_1+\epsilon r_2
< \Delta$. From Strong Law of Large Numbers, there exist time $T$
such that for $t>T$,
\begin{align}
\widehat{P}_t(C_i\geq r_1 \big |\widehat{C}_i=\hat{c}_k )
-P(C_i\geq r_1 \big |\widehat{C}_i=\hat{c}_k ) &>-\delta \quad a.s., \nonumber\\
\widehat{P}_t(C_i\geq r_2 \big |\widehat{C}_i=\hat{c}_k ) -P(C_i\geq
r_2 \big |\widehat{C}_i=\hat{c}_k ) &<\epsilon \quad a.s. \nonumber
\end{align}

Hence we have
\begin{align}
& \widehat{P}_t(C_i\geq r_1 \big |\widehat{C}_i=\hat{c}_k )r_1 -
\widehat{P}_t(C_i\geq r_2 \big |\widehat{C}_i=\hat{c}_k )r_2\nonumber \\
> & P(C_i\geq r_1 \big |\widehat{C}_i=\hat{c}_k )r_1-\delta r_1 -
P(C_i\geq r_2 \big |\widehat{C}_i=\hat{c}_k ) r_2 - \epsilon r_2 \nonumber \\
= &P(C_i\geq r_1 \big |\widehat{C}_i=\hat{c}_k )r_1-
P(C_i\geq r_2 \big |\widehat{C}_i=\hat{c}_k ) r_2- (\delta r_1 + \epsilon r_2) \nonumber \\
> & \Delta-\Delta =0 \quad a.s. \nonumber
\end{align}
\end{proof}

Remark: This lemma implies that, there will be a time beyond which
the allocated rate with empirical knowledge of
$\widehat{P}_t(C_i\geq r_1 \big |\widehat{C}_i=\hat{c}_k )$ is the
same as with accurate knowledge. Because both the number of users
and the state space is finite, we will chose the right rate with
probability one as time is large, which is summarized in the
following corollary.

\begin{corollary}
There exist a time $\widetilde{T}$ beyond which, with probability 1, the empirical
scheme will allocate rate $\widehat{R}[t]$ the same as $R[t]$ when
the $P(C[t] | \widehat{C}[t])$ is perfect known.
\end{corollary}

\begin{proof}
This result is immediate from the previous lemma.
\end{proof}

\begin{lemma}
At those non-observation slot $t$, let $\widehat{I}[t]$ and
$\widehat{R}_i[t]$ be the scheduling decision by the joint statistics learning-scheduling policy. Let $I[t]$ and $R_i[t]$ be the scheduling decision of $\Psi$ with
accurate knowledge.for any $\rho>0$. There exist time $N$ beyond which, with probability 1,
\begin{align}
&\Big |\sum_{j=1}^N Q_i[t] \bm 1(I[t]{=}j) P(C_j[t]{\geq} R_i[t] \big
|\widehat{C}_j[t]{=}\hat{c}_j )R_i[t]{-}\sum_{i=1}^N Q_i[t] \bm
1(\widehat{I}[t]{=}i) P(C_i[t]{\geq} \widehat{R}_i[t] \big
|\widehat{C}_i[t]{=}\hat{c}_i )\widehat{R}_i[t] \big | \nonumber
\\
\leq & \rho \sum_{i=1}^N Q_i[t] \nonumber
\end{align}
\end{lemma}

\begin{proof}
From corollary 9, let $T_1$ be the time beyond which $R[t]=\widehat{R}[t]$ a.s.. We
consider the time $t>T_1$ and thus $R_i[t]=\widehat{R}_i[t]=r_{\hat{c}_i}^*= \arg\max P(C_j[t]{\geq} r \big |\widehat{C}_j[t]{=}\hat{c}_j ) \cdot r$ almost surely for all $i$. From strong law of large numbers, let $T_2\geq T_1$ be such that beyond which $| \widehat{P}_t(C_j{\geq}
r_{\hat{c}_j}^* \big |\widehat{C}_j{=}\hat{c}_j )r_{\hat{c}_j}^* -
P(C_j{\geq} r_{\hat{c}_j}^* \big |\widehat{C}_j{=}\hat{c}_j
)r_{\hat{c}_j}^* | \leq \rho$ almost surely for all $j$. We henceforth consider $t>T_2$. Given queue length $\bm Q[t]$ and estimated channel state information $\widehat{\bm C}_i[t]$, $I[t]$ and $\widehat{I}[t]$ are determined by
\begin{align}
I[t]=&\argmax_i \ Q_i[t] P(C_i\geq r_{\hat{c}_i}^* \big
|\widehat{C}_i=\hat{c}_i
)r_{\hat{c}_i}^* \quad a.s. \nonumber \\
\widehat{I}[t]=&\argmax_i \ Q_i[t] \widehat{P}_t(C_i\geq r_{\hat{c}_i}^*
\big |\widehat{C}_i=\hat{c}_i )r_{\hat{c}_i}^* \quad a.s. \nonumber
\end{align}

If $I[t]=\widehat{I}[t]$, the statement will hold. If $I[t]=h$ but $\widehat{I}[t]=k$ for $h\neq k$, we have
\begin{align}
Q_k[t] \widehat{P}_t(C_k{\geq} r_{\hat{c}_k}^* \Big
|\widehat{C}_k{=}\hat{c}_k)r_{\hat{c}_k}^* &{\geq}  Q_h[t] \widehat{P}_t(C_h{\geq} r_{\hat{c}_h}^* \Big
|\widehat{C}_h{=}\hat{c}_h)r_{\hat{c}_h}^*,\\
Q_h[t] {P}(C_h{\geq} r_{\hat{c}_k}^* \Big
|\widehat{C}_h{=}\hat{c}_h)r_{\hat{c}_h}^* &{\geq}  Q_k[t] {P}(C_k{\geq} r_{\hat{c}_k}^* \Big
|\widehat{C}_k{=}\hat{c}_k)r_{\hat{c}_k}^*.
\end{align}

Because $t\geq T_2$, we further have
\begin{align}
&\sum_{j=1}^N Q_i[t] \bm 1(I[t]{=}j) P(C_j{\geq} R_i[t] \big
|\widehat{C}_j{=}\hat{c}_j )R_i[t]{-}\sum_{i=1}^N Q_i[t] \bm
1(\widehat{I}[t]{=}i) P(C_i{\geq} \widehat{R}_i[t] \big
|\widehat{C}_i{=}\hat{c}_i )\widehat{R}_i[t]\nonumber \\
=&Q_h[t] P(C_h{\geq} r_{c_h}^* \big
|\widehat{C}_h{=}\hat{c}_h)r_{c_h}^* -
Q_k[t] P(C_k{\geq} r_{c_k}^* \big |\widehat{C}_k{=}\hat{c}_k)r_{c_k}^* \nonumber\\
\leq & Q_h[t] [\widehat{P}_t(C_h{\geq} r_{c_h}^* \big
|\widehat{C}_h{=}\hat{c}_h)r_{c_h}^*+\rho ]-
Q_k[t] [\widehat{P}_t(C_k{\geq} r_{c_k}^* \big |\widehat{C}_k{=}\hat{c}_k)r_{c_k}^*-\rho]\nonumber\\
\leq & \rho (Q_h[t]+Q_k[t]) \nonumber
\end{align}
almost surely, where the first inequity is from the assumption that $| \widehat{P}_t(C_j{\geq}
r_{\hat{c}_j}^* \Big |\widehat{C}_j{=}\hat{c}_j )r_{\hat{c}_j}^* -
P(C_j{\geq} r_{\hat{c}_j}^* \Big |\widehat{C}_j{=}\hat{c}_j
)r_{\hat{c}_j}^* | \leq \rho$ almost surely for all $j$, and the last inequality is from (13). Also, we have
\begin{align}
&\sum_{i=1}^N Q_i[t] \bm
1(\widehat{I}[t]{=}i) P(C_i{\geq} \widehat{R}_i[t] \Big
|\widehat{C}_i{=}\hat{c}_i )\widehat{R}_i[t]-\sum_{j=1}^N Q_i[t] \bm 1(I[t]{=}j) P(C_j{\geq} R_i[t] \Big |\widehat{C}_j{=}\hat{c}_j )R_i[t] \nonumber \\
=&Q_k[t] P(C_k{\geq} r_{c_k}^* \Big |\widehat{C}_k{=}\hat{c}_k)r_{c_k}^*-Q_h[t] P(C_h{\geq} r_{c_h}^* \Big |\widehat{C}_h{=}\hat{c}_h)r_{c_h}^* \nonumber\\
\leq & 0 \nonumber\\
\leq & \rho (Q_h[t]+Q_k[t]) \nonumber
\end{align}
almost surely, where the first inequality is from (14). We hence proved the Lemma.
\end{proof}

\noindent \textbf{(Proof of Proposition 5)}
\begin{proof}
We first prove that for any $\bm \lambda$ strictly within $\bm \Lambda_{\gamma}$, i.e., $\bm {\lambda}+\varepsilon \vec{\bm 1} \in int (\bm
\Lambda_{\gamma})$, there exists a policy $G_0$ that makes decision only based on empirical statistics, i.e., $\widehat{P}_t(C_i\geq r | \widehat{C}_i)$ and stably supports $\lambda$. Beccause
\begin{align}
\lambda_i +\varepsilon \leq (1-\gamma)
\sum_{\hat{\bm c}\in \mathcal{S}^N} \hat{\pi}_{\hat{\bm c}} \cdot
\alpha_i^{\hat{\bm c}} \cdot P(C_i\geq r_i^*(\hat{c}_i) \big
|\widehat{C}_i=\hat{c}_i\ ) \cdot r_i^*(\hat{c}_i).
\end{align}

Then the
randomized policy $G_0$ can be: at every non-observing estimated state
$\hat{c}$, activates channel $\widetilde{I}[t]=i$ with probability
$\alpha_i^{\hat{\bm c}}$ with the allocated transmission rate
$\widetilde{R}_i[t]=\arg\max_r \widehat{P}_t(C[t]\geq r | \widehat{C}[t])\cdot r$. Define Lyapunov function $L(\bm Q[t])=\sum_{i=1}^{N} Q^2_i[t]$, similar to (7), the Lyapunov drift can be written as
\begin{align}
\nonumber \Delta L[t]&=E \Big[ \sum_{i=1}^{N} Q_i^2[t+1]-Q_i^2[t]
\Big | \ \bm Q[t] \ \Big ]\\
& \leq B_1+ 2\sum_{i=1}^{N} Q_i[t] \Big(\lambda_i -
(1-\gamma) E\Big[\bm 1
(\widetilde{I}[t]=i) \cdot \widetilde{R}[t] \cdot \bm 1(\widetilde{R}[t] \leq C_i[t]) \Big | \ \bm Q[t] \ \Big ] \ \Big)
\end{align}
where $B_1$ is bounded,
\begin{displaymath}
B_1 = (1-\gamma)^2 \sum_{i=1}^{N} E\big[\bm 1 (\widetilde{I}[t]=i) \cdot
\widetilde{R}^2[t] \bm 1 (\widetilde{R}[t] \leq C_i[t])+ A_i^2[t] \Big |
\ \bm Q[t] \big].
\end{displaymath}

Then from Corollary 9, for $t>N$, $\widetilde{R}[t]=r_i^*(\hat{c}_i)=\arg\max_r P(C_i\geq r
\big |\widehat{C}_i=\hat{c}_i\ ) \cdot r$ almost surely. Hence
\begin{align}
&(1-\gamma)E\Big[\bm 1 (\widetilde{I}[t]=i) \cdot \widetilde{R}[t] \cdot \bm
1(\widetilde{R}[t] \leq C_i[t]) \Big | \ \bm Q[t] \ \Big ] \
\Big) \nonumber\\
=&(1-\gamma)\sum_{\hat{\bm c}\in \mathcal{S}^N} \hat{\pi}_{\hat{\bm
c}} \cdot \alpha_i^{\hat{\bm c}} \cdot P(C_i\geq r_i^*(\hat{c}_i)
\big |\widehat{C}_i=\hat{c}_i\ ) \cdot r_i^*(\hat{c}_i)
\end{align}

Substitute (15) and (17) into equation (16),the Lyapunov drift function will take the form
\begin{align}
\nonumber \Delta L[t]\leq B- 2 \varepsilon \sum_{i=1}^{N} Q_i[t].
\end{align}
Hence Queues will be stable.

We next show that the joint statistics learning-scheduling policy will stabilize $\bm \lambda$ similar to the proof of Proposition 2. Given queue length information $\bm Q[t]$ and estimated channel state $\widehat{C}[t]$, suppose the proposed joint statistics learning-scheduling policy will result in rate adaptation $\widehat{R}[t]$ and scheduling decision $\widehat{I}[t]$ at time $t$, and suppose the policy $\Psi$ with perfect CSI will make rate adaptation decision $R[t]$ and scheduling decision $I[t]$. Associated with the joint statistics learning and scheduling algorithm, the Lyapunov Drift can be written as
\begin{align}
\nonumber \Delta L[t]&=E \Big[ \sum_{i=1}^{N} Q_i^2[t+1]-Q_i^2[t]
\Big | \ \bm Q[t] \ \Big ]\\
& \leq B_2+ 2\sum_{i=1}^{N} Q_i[t] \Big(\lambda_i -
(1-\gamma) E\Big[\bm 1
(\widehat{I}[t]=i) \cdot \widehat{R}[t] \cdot \bm 1(\widehat{R}[t] \leq C_i[t]) \Big | \ \bm Q[t] \ \Big ] \ \Big) \nonumber \\
& = B_2+2\sum_{i=1}^{N} Q_i[t] \Big(\lambda_i -(1-\gamma)
E\Big[ \ E\Big[\bm 1 (\widehat{I}[t]=i) \cdot \widehat{R}[t] \cdot \bm
1(\widehat{R}[t] \leq C_i[t]) \Big | \ \bm Q[t], \ \widehat{\bm
C}[t] \Big ] \ \Big ] \Big)
\end{align}
where $B_2$ is bounded,
\begin{displaymath}
B_2 = (1-\gamma)^2 \sum_{i=1}^{N} E\big[\bm 1 (\widehat{I}[t]=i) \cdot
\widehat{R}^2[t] \bm 1 (\widehat{R}[t] \leq C_i[t])+ A_i^2[t] \Big |
\ \bm Q[t] \big].
\end{displaymath}

Same as in proof of Proposition 2, we have
\begin{align}
\nonumber & \sum_{i=1}^{N} Q_i[t] \cdot E\Big[\bm 1 (\widetilde{I}[t]=i) \cdot
\widetilde{R}[t] \cdot \bm 1(\widetilde{R}[t] \leq C_i[t]) \ \Big ]\\
\nonumber \leq & \sum_{i=1}^{N} Q_i[t] \cdot E\Big[ E\Big[\bm 1 ({I}[t]=i) \cdot
{R}[t] \cdot \bm 1({R}[t] \leq C_i[t])\Big | \ \bm Q[t], \ \widehat{\bm
C}[t] \Big ]  \ \Big].
\end{align}

And therefore for any $0<\rho<\varepsilon/(1-\gamma)$, for $t>\max \{T_1, N\}$, we have
\begin{align}
\sum_{i=1}^{N} Q_i[t] \cdot (\lambda_i+\varepsilon) & \leq
(1-\gamma)\sum_{i=1}^{N} Q_i[t] \cdot E\Big[\bm 1 (\widetilde{I}[t]=i) \cdot
\widetilde{R}_{i}[t] \cdot \bm
1(\widetilde{R}_{i}[t]\leq C_i[t])\Big]\nonumber\\
&\leq (1-\gamma)\sum_{i=1}^{N} Q_i[t] \cdot E\Big[ E\Big[\bm 1 ({I}[t]=i)
\cdot {R}[t] \cdot \bm 1({R}[t] \leq C_i[t]) \Big | \ \bm Q[t], \
\widehat{\bm C}[t]\Big ] \Big ] \nonumber\\
&= (1-\gamma)\sum_{i=1}^{N} Q_i[t] \cdot E\Big[E\Big[\bm 1 ({I}[t]=i)
\cdot \widehat{R}[t] \cdot \bm 1(\widehat{R}[t] \leq C_i[t]) \Big | \ \bm Q[t], \
\widehat{\bm C}[t]\Big ] \Big ]\nonumber\\
&=(1-\gamma)\sum_{i=1}^{N} Q_i[t] \sum_{C_i[t] \in \mathcal{S}} \bm 1 ({I}[t]=i)
\cdot \widehat{R}[t] \cdot P(\widehat{R}[t] \leq C_i[t] \Big | \widehat{C}_i[t])\nonumber\\
&=(1-\gamma)\sum_{C_i[t] \in \mathcal{S}} \sum_{i=1}^{N} Q_i[t]  \bm 1 ({I}[t]=i)
\cdot \widehat{R}[t] \cdot P(\widehat{R}[t] \leq C_i[t] \Big | \widehat{C}_i[t])\nonumber\\
&\leq(1-\gamma)\sum_{C_i[t] \in \mathcal{S}} \sum_{i=1}^{N} Q_i[t]  \Big[\bm 1 ({I}[t]=i)
\cdot \widehat{R}[t] \cdot P(\widehat{R}[t] \leq C_i[t] \Big | \widehat{C}_i[t])+\rho \Big]\nonumber\\
&= (1-\gamma)\sum_{i=1}^{N} Q_i[t] \cdot E\Big[ E\Big[\bm 1 (\widehat{I}[t]=i)
\cdot \widehat{R}[t] \cdot \bm 1(\widehat{R}[t] \leq C_i[t]) \Big | \ \bm Q[t], \
\widehat{\bm C}[t]\Big ] + \rho \Big ]
\end{align}
where the last inequality comes from Lemma 10.
Substitute (19) into the Lyapunov drift expression (18), we will have:
\begin{align}
\nonumber \Delta L[t]& \leq B_2- 2 \big(\varepsilon-\rho(1-\gamma)\big) \sum_{i=1}^{N} Q_i[t]
\end{align}
Since $\varepsilon-\rho(1-\gamma)>0$, the queues will be stable.
\end{proof}


\begin{thebibliography}{1}

\bibitem{backpressure}
L. Tassiulas, A. Ephremides, ``Stability properties of constrained
queueing systems and scheduling policies for maximum throughput in
multihop radio networks,'' \emph{IEEE Transactions on Automatic
Control,} vol. 37, no. 12, pp. 1936-1948, Dec. 1992.

\bibitem{MWM}
L. Tassiulas, A. Ephremides, ``Dynamic server allocation to parallel
queues with randomly varying connectivity,'' \emph{IEEE Transactions
on Information Theory,} vol. 39, no. 2, pp. 466-478, Mar. 1993.

\bibitem{NessLin05}
X. Lin, N. B. Shroff. ``Joint rate control and scheduling in
multihop wireless networks,'' \emph{Proceedings of IEEE Conference
on Decision and Control,} Paradise Island, Bahamas, Dec. 2004.

\bibitem{ImpftSch}
X. Lin, N. B. Shroff, ``The impact of imperfect scheduling on
cross-Layer congestion control in wireless networks,''
\emph{IEEE/ACM Transaction on Networking,} vol. 14, no. 2, pp.
302-315, Apr. 2006.

\bibitem{Eryilmaz05}
A. Eryilmaz, R. Srikant, ``Fair resource allocation in wireless
networks using queue-length based scheduling and congestion
control,'' \emph{IEEE/ACM Transaction on Networking,} vol. 15, no.
6, pp. 1333-1344, Dec. 2007.

\bibitem{EryilmazJsac}
A. Eryilmaz, R. Srikant, ``Joint congestion control, routing and MAC
for stability and fairness in wireless networks,'' \emph{IEEE
Journal on Selected Areas in Communications,} vol. 24, no. 8, pp.
1514-1524, Aug. 2006.

\bibitem{Sasha}
A. Stolyar, ``Maximizing queueing network utility subject to
stability: greedy primal-dual algorithm,'' \emph{Queueing Systems,}
vol. 50, no.4, pp.401-457, Aug. 2005.

\bibitem{Neely05}
M. J. Neely, E. Modiano, C. Li, ``Fairness and optimal stochastic
control for heterogeneous networks,'' \emph{IEEE Transactions on
Information Theory,} vol. 52, no. 7, pp. 2915-2934, Jul. 2006.

\bibitem{infreq}
K. Kar, X. Luo, S. Sarkar, ``Throughput-optimal scheduling in
multichannel access point networks under infrequent channel
measurements,'' \emph{IEEE Transactions on Wireless Communications,}
vol. 7, no. 7, pp. 2619-2629, Jul. 2008,

\bibitem{Topology}
L. Ying, S. Shakkottai, ``Scheduling in mobile Ad Hoc networks with
topology and channel-State Uncertainty,'' \emph{Proceedings of IEEE
INFOCOM,} Rio de Janeiro, Brazil, Apr. 2009.

\bibitem{PartialSet}
A. Gopalan, C. Caramanis, S. Shakkottai,``On wireless scheduling
with partial channel-state information,'' \emph{Proceedings of
Allerton Conference on Communication, Control, and Computing,}
Monticello, IL, Sept. 2007.


\bibitem{Neely_CDC}
C. Li, M. J. Neely, ``Energy-optimal scheduling with dynamic channel
acquisition in wireless downlinks,'' \emph{Proceedings of IEEE
Conference on Decision and Control,} New Orleans, LA, Dec. 2007.

\bibitem{Uncertainty}
A. Pantelidou, A. Ephremides, A.L. Tits, ``Joint scheduling and
routing for Ad-hoc networks under channel state uncertainty,''
\emph{Proceedings of IEEE Intl. Symp. on Modeling and Optimization in Mobile, Ad
Hoc, and Wireless Networks (WiOpt),} Limassol, Cyprus, Apr. 2007.

\bibitem{Phil}
R. Aggarwal, P. Schniter, C. E. Koksal, ``Rate adaptation via
link-layer feedback for goodput maximization over a time-varying
channel,'' \emph{IEEE Transactions on Wireless Communications,} vol.
8, no. 8, pp. 4276-4285, Aug. 2009.

\bibitem{2stage}
M. J. Neely, ``Max weight learning algorithms with application to
scheduling in unknown environments,'' \emph{arXiv:0902.0630v1,} Feb.
2009.

\bibitem{Boyd}
S. Boyd, L. Vandenberghe, \emph{``Convex optimization,''} Cambridge
University Press, 2004.


\bibitem{Separation}
D. Bertsekas, \emph{``Nonlinear Programming,''} Belmont, MA: Athena Scientific, 1995.

\bibitem{Sean}
S. Meyn, \emph{``Control techniques for complex networks,''}
Cambridge University Press, 2007.

\bibitem{Billingsley}
P. Billingsley, \emph{``Probability and measure, 3rd edition,''}
Wiley, New York, 1995.

\bibitem{Gallager}
R. G. Gallager, \emph{``Discrete stochastic processes,''} Boston,
MA: Kluwer, 1996.

\bibitem{Tse}
D. Tse, P. Viswanath, \emph{``Fundamentals of wireless
communication,''} Cambridge University Press, 2005.

\bibitem{Junshan1}
D. Zheng, S. Pun, W. Ge, J. Zhang, H.V. Poor, ``Distributed
opportunistic scheduling for Ad-Hoc communications under noisy
channel estimation,'' \emph{Proceedings of IEEE ICC,} Beijing, China
May 2008.

\bibitem{Hassibi}
A. Vakili, M. Sharif, B. Hassibi, ``The effect of channel estimation
error on the throughput of broadcast channels,'' \emph{Proceedings of IEEE
Int. Conf. Acoust. Speech Signal Processing,} Toulouse, France, May
2006.

\bibitem{Hassibi2}
T. Kailath, A.H. Sayed, B. Hassibi, \emph{``Linear estimation,''}
Prentice-Hall, 2000.

\end{thebibliography}
\end{document}